\documentclass[preprint,12pt]{elsarticle}
\usepackage{eucal,  amssymb,amsmath,amsthm,mathrsfs, amsfonts, latexsym, subfigure}
\usepackage[T1]{fontenc}
\usepackage{float} 

\newtheorem{theorem}{Theorem}[section]

\newtheorem{remark}[theorem]{Remark}
\newtheorem{Corollary}[theorem]{Corollary}
\numberwithin{equation}{section}

\IfFileExists{upquote.sty}{\usepackage{upquote}}{}
\begin{document} \parskip=5pt plus1pt minus1pt \parindent=0pt
\begin{frontmatter}
\title{SEIRS epidemics in growing populations}
\author[su]{Tom~Britton}
\ead{tom.britton@math.su.se}
\author[uo1]{D\'esir\'e~Ou\'edraogo\corref{cor1}}
\ead{dsrdrg@gmail.com}
\cortext[cor1]{Corresponding author}
\address[su]{Department of Mathematics, Stockholm University, SE-106 91 Stockholm, Sweden}
\address[uo1]{Laboratoire de Math\'{e}matiques et Informatique (LAMI), EDST,  Universit\'e Ouaga I Pr. Joseph Ki-Zerbo 03 B.P.7021 Ouagadougou 03, Burkina Faso}
\date{\today}

\begin{abstract} 
An SEIRS epidemic with disease fatalities is introduced in a growing population (modelled as a super-critical linear birth and death process). The study of the initial phase of the epidemic is stochastic, while the analysis of the major outbreaks is deterministic. Depending on the values of the parameters, the following scenarios are possible. i) The disease dies out quickly, only infecting few; ii) the epidemic takes off, the \textit{number} of  infected individuals grows exponentially, but the \textit{fraction} of infected individuals remains negligible; iii) the epidemic takes off,  the \textit{number} of infected grows initially quicker than the population, the disease fatalities diminish the growth rate of the population, but it remains super critical, and the \emph{fraction} of infected go to an endemic equilibrium; iv) the epidemic takes off, the \textit{number} of infected individuals grows initially quicker than the population, the diseases fatalities turn the exponential growth of the population to an exponential decay. 
 \end{abstract}
\end{frontmatter}
\section{Introduction}\label{sec-intro}
Infectious diseases remain a threat for developing countries as well as for developed countries. Many mathematicians focus their efforts to understand the dynamics of infectious diseases, in order to find the conditions to eradicate them. In mathematical modelling of infectious disease epidemics, the population in which the disease is spreading is partitioned in several compartments according to the status of the individuals, related to the disease. Every epidemic model has at least, the compartment $I$ of the infectious individuals who are infected and able to transmit the disease to others through contact, and the compartment $S$  of the susceptible individuals (those who are not infected but may be infected if they contact an infectious individual).  Two other compartments often used are the compartment  $E$ of the exposed or latent individuals who are already infected but not yet able to transmit the disease to others, and the compartment $R$ of the recovered or removed individuals (those who are healed from the disease with a permanent or non-permanent immunity). In a $SEIR$ epidemic, a susceptible individual infected through a contact with an infectious, becomes infected and latent; at the end of the latent period he/she becomes infectious and at the end of the infectious period he/she recovers with a life-long immunity. An $SEIRS$ epidemic is almost the same as the preceding, the only difference is that a recovered individual has a non-permanent immunity (He/she can lose his immunity, becoming susceptible again). Diphtheria, influenza and pneumonia are examples of diseases with latent period and non-permanent immunity \cite{Greenhalgh}. 

In \cite{Tom},  Britton and  Trapman studied stochastic SIR and SEIR models in a growing population. They derived the basic reproduction number and the Malthusian parameter of the epidemic,  stated results for the initial phase and showed that the stochastic proportions process converges to a deterministic process.

In \cite{Greenhalgh}, Greenhalgh studied an SEIRS deterministic model with vaccination and found that under some conditions, the solution has Hopf bifurcations. 

The aim of this paper is to study the dynamic of a stochastic SEIRS epidemic model with disease induced mortality, in an exponentially growing population. As in \cite{Tom} , we assume that  without the disease, the population has a  birth rate $\lambda$, and a natural death rate $\mu$, such that $\lambda>\mu$. That is, initially the population process is a super-critical linear birth and death process. An SEIRS epidemic is introduced by infecting one individual.  With the disease, the population is divided in four compartments according to the status of the individuals, related to the disease. The compartment S of the susceptible individuals, the compartment E of the latent or exposed individuals , the compartment I of the infectious individuals , and the compartment R of the removed individuals (those who are healed of the disease with a non-permanent immunity). The process is initiated by setting $(S(0),E(0),I(0),R(0))=(n-1,1,0,0)$. The transfer diagram of the model is given by Figure \ref{SEIRSdiagram}.

We derive the Malthusian parameter $\alpha$, the basic reproduction number $R_0$ and the probability of minor outbreak $\pi$ of the epidemic. If $R_0\leq1$, then the disease cannot invade the population, that remains a super critical process. If $R_0>1$, then the epidemic has a positive probability $1-\pi$ of taking off, with the remaining probability $\pi$, it dies out. If the epidemic takes off, another threshold parameter $R_1$ determines the behavior of the proportion of infected individuals. If $R_1\leq1$, then the \textit{fraction} of infected stays small; while it persists when $R_1>1$. If $R_1>1$, or equivalently $\alpha>\lambda-\mu$, then the \textit{number} of the infected grows initially quicker than the population, the disease fatalities diminish the growth rate of the population. In this case the asymptotic behavior of the population rely on a third threshold quantity $R_2$. If $R_2>1$, then the population goes on growing, while it becomes a sub-critical process when $R_2\leq1$. In the latter case, when the number of individuals become low, the population should vanish with the disease, or regrows after the extinction of the epidemic.  
 
 We start by defining the stochastic model in Section $2$. Then, in Section $3$, we study the initial phase of the epidemic, thereafter we consider the deterministic model in Section $4$. Afterward,  we give some illustrations by simulating different scenarios of epidemics in Section $5$. In Section $6$, we conclude the paper and discuss some perspectives.  
\section{The model}\label{sec-models}
\subsection{The initial dynamic of the population}
Initially (before the introduction of the disease), the population model is a linear birth and death (B-D) process with individual birth rate $\lambda$ and individual death rate $\mu$. We assume that $\lambda>\mu$, that is the process is super-critical. $N(t)$ denotes the \textit{number} of individuals in the population at time $t$. 
\begin{figure}[t]
\centering
\includegraphics[height=7 cm, width=12cm]{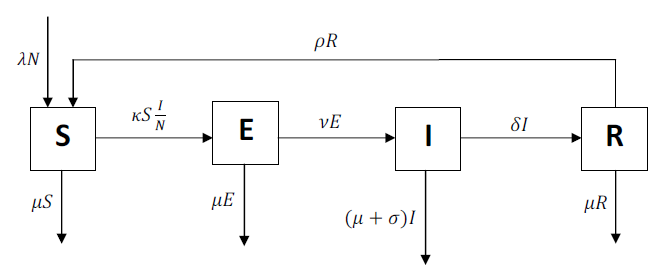}
\caption{The transfer diagram of the SEIRS model}\label{SEIRSdiagram}
\end{figure}
\subsection{The Markovian SEIRS epidemic model}
Now, we define a uniformly mixing Markovian epidemic model on the population described above, implying that individuals give birth at rate $\lambda$ and die from other causes at rate $\mu$, irrespective of disease status.
Initially (at $t=0$), the population consists of $n$ susceptible individuals. At this time, an SEIRS infectious  disease is introduced by infecting one individual. The disease spreading is modelled as follow. An infected individual remains latent (infected but not yet infectious) for exponential time with rate $\nu$. After this period, he/she becomes infectious (unless he/she dies before). An infectious individual remains infectious for an exponential time with rate $\delta$ (unless he/she dies before). The disease  induces an additional death rate  $\sigma$ for the infectious individuals. After the period of infectiousness, an infectious individual recovers with a temporary immunity. A recovered person remains immune for an exponential time with rate $\rho$ unless he/she dies before. After the period of immunity he/she becomes susceptible again. During the infectious period, the infective has infectious contacts randomly in time according to a homogeneous Poisson process with rate $\kappa$, each time with a uniformly selected random individual. If the contacted person is susceptible, he/she becomes infected and latent (not yet infectious), otherwise the contact has no effect. There is no vertical transmission, that is all the new born individuals are susceptible.

Let $Z(t) = (S(t),E(t),I(t),R(t))$ respectively denote the \textit{number} of susceptible, latent, infectious and immune individuals at time $t$. Therefore, the population size at time $t$ is $N(t)=S(t)+E(t)+I(t)+R(t)$  (Wherever $n$ is important an $n$-index is added). The population is initiated at $Z_n(0) = (S_n(0),E_n(0),I_n(0),R_n(0)) = (n-1,1,0,0)$. However, later we  also derive the probability of minor outbreak for an epidemic starting with $m$ latent individuals and $j$ infectious individuals with a very small infected \textit{fraction} ($m+j\ll n$). 

In short, the population model has two parameters, the birth rate $\lambda$ and the death rate $\mu$; and the disease model  has five parameters, the transmission rate $\kappa$, the end of latency rate  $\nu$, the recovery rate $\delta$, the disease death rate $\sigma$ and the immunity waning rate $\rho$. The possible events and their rates, when currently in state $Z(t)=(S(t),E(t),I(t),R(t))=(u,v,x,y)=z$, are given in Table \ref{notations}.

The SI, SIS, SIR, SIRS, SEI, SEIS, SEIR models are special cases of the SEIRS model defined above. If $\rho=0$, then an infected individual cannot go back to the susceptible state and we get an SEIR model. If $\rho\longrightarrow\infty$, then the recovered state vanishes and we get an SEIS model. If $\nu\longrightarrow\infty$, then the latent state vanishes giving an SIRS model. If $\delta=0$, then we have lifelong infectivity and hence an SEI model. If $\nu\longrightarrow\infty$ and $\rho\longrightarrow\infty$, then the latent state and the recovered state vanish and we have an SIS model. If $\nu\longrightarrow\infty$ and $\delta=\rho=0$,  then the latent state vanishes and we have lifelong infectivity, giving hence an SI model. Therefore from the results of the SEIRS model, one can deduce those of the others.
\begin{table}[h]
\begin{centering}
\begin{tabular}{|l|l|l|l|} 
\hline
Event& State change $l$&Rate\\ \hline
Birth& $(1,0,0,0)$&$\lambda N$\\ \hline
Death of susceptible& $(-1,0,0,0)$&$\mu S$\\ \hline
Death of exposed& $(0,-1,0,0)$&$\mu E$\\ \hline
Death of infective& $(0,0,-1,0)$&$(\mu+\sigma)I$\\ \hline
Death of recovered & $(0,0,0,-1)$&$\mu R$\\ \hline
Infection& $(-1,1,0,0)$&$\kappa SI/N$ \\ \hline
End of latency& $(0,-1,1,0)$&$\nu E$  \\ \hline
Recovery&$(0,0,-1,1)$&$\delta I$\\ \hline
Loss of immunity&$(1,0,0,-1)$&$\rho R$\\ \hline
\end{tabular}
\caption{The uniform Markovian dynamic SEIRS epidemic: type of events, their state change $l$ (The old state $z=(u,v,x,y)$ is hence changed to $z+l$) and their rates.}\label{notations}
\end{centering}
\end{table}

In this section we have presented the stochastic SEIRS model studied in this paper.
\section{Results for the initial phase}
\subsection{The dynamics of the population size N(t)}
Without the disease, the population process is a linear super-critical birth and death process with individual birth rate $\lambda$ and individual death rate $\mu$. But, with the introduction of the disease that  induces an extra death rate $\sigma$ for the infectious individuals, we have two possible events. Birth with rate $\lambda N(t)$ and death with rate $\mu N(t)+\sigma I(t)=[\mu+\sigma I(t)/N(t)]N(t)$, where $N(t)$ is the current total \textit{number} of individuals in the population and $I(t)$ the current \textit{number} of infectious individuals. Then, the population process is no longer a linear birth and death process (unless $\sigma = 0$). However, at the beginning of the epidemic when the \textit{fraction} of the infectious individuals is very small ($I(t)/N(t)\approx0$), the population will behave almost as a linear birth and death process with birth rate $\lambda$ and death rate $\mu$. Thus for the initial phase of the epidemic, we assume that the population size is a  linear super-critical birth and death process with birth rate $\lambda$ and death rate $\mu$. On the other hand, if the epidemic takes off after the initial phase, whenever the \textit{fraction} of the infectious individuals remains below $(\lambda-\mu)/\sigma$, the population process will be a super-critical process, having a positive probability to grow to $\infty$.  But, if the infectious \textit{fraction} grows beyond $(\lambda-\mu)/\sigma$, then the population process  becomes a sub-critical process.
\begin{remark}
If $(\lambda-\mu)/\sigma\geq1$ i.e. $\lambda\geq\mu+\sigma$, then N(t) is always a super-critical process although the epidemic.
\end{remark}
\subsection{Approximation of the initial phase of the epidemic}
Now, we consider the beginning of the epidemic, when the \textit{fraction} of the infected individuals is still small. $E(t)$ increases with rate $\kappa I(t)S(t)/N(t)$ due to infection and decreases by rate $(\nu+\mu)E(t)$ due to death or end of latency. $I(t)$ increases with rate $\nu E(t)$ due to end of latency, and decreases with rate $(\delta+\mu+\sigma)I(t)$ due to recovery or death. As we start with  $n-1$ susceptible and one infected, assuming that $n$ is large, at the beginning of the epidemic $S(t)/N(t)$ is very close to $1$. Then, the \textit{number} of exposed individuals increases almost with rate $\kappa I(t)$. Let $L_n(t)=E_n(t)+I_n(t)$ be the \textit{number} of infected (latent or infectious) individuals at time $t$. $L_n(t)$ can be approximated by a branching process, $L_\infty(t)=E_\infty(t)+I_\infty(t)$  with two stages: the childhood (or latent) stage $E_\infty$ and the adult (or infectious) stage $I_\infty$ \cite[p. 54]{Odo}. Initiated at $(E_\infty(0),I_\infty(0))=(1,0)$, where $E_\infty(t)$ increases with rate $\kappa I_\infty(t)$ and decreases with rate $(\mu+\nu)E_\infty(t)$, and $I_\infty(t)$ increases with rate $\nu E_\infty(t)$ (end of childhood) and decreases with rate $(\delta+\mu+\sigma)I_\infty(t)$ (end of adult stage). 
\begin{theorem}
Let $L_n(t)$ be the epidemic process and $L_\infty(t)$ be the branching process defined above. Then, $L_n(t)$ converges weakly to $L_\infty(t)$ ($L_n\stackrel{w}{\Longrightarrow}L_\infty$), as $n\longrightarrow\infty$,  on any finite interval $[0,t_1]$.
\end{theorem}
\begin{proof}
Like initial phase of other epidemics \cite[p. 54]{Odo}, when $n$ tends to infinity, the transition probabilities of the epidemic converge to that of the branching process.
\end{proof}
The results below are for the branching process $L_\infty$, but since $L_n\stackrel{w}{\Longrightarrow}L_\infty$, they apply to the epidemic as $n\longrightarrow\infty$.
\subsection{Thresholds}
In this subsection, we derive the Malthusian parameter $\alpha$ and the basic reproduction number $R_0$ of the limiting branching process $L_\infty$.

The Malthusian parameter $\alpha$ is defined as the exponential growth/decay rate the epidemic branching process has. It is the solution of
\begin{equation}
\int_0 ^\infty e ^{-\alpha t}c(t)dt=1,\label{c}
\end{equation}
where $c(t)$ is the expected rate at which an individual gives birth (has infectious contacts) $t$ time units after it was infected~\cite[page 10]{Jagers}.
\begin{theorem}
The Malthusian parameter of the epidemic is given by
\begin{equation}
\alpha=\displaystyle -\left(\mu+\frac{\nu+\delta+\sigma}{2}\right)+\sqrt{\frac{(\nu-\delta-\sigma)^2}{4}+\kappa \nu}.\label{alpha}
\end{equation}
\end{theorem}
\begin{proof}
 In the SEIRS model, the contact rate is $0$ during the latent period and $\kappa$ during the infectious period. By conditioning on when the latent period ends, it follows that 
\[c(t)=\kappa e^{-\mu t}\int_0 ^t \nu e^{-\nu s}e^{-(\delta+\sigma)(t-s)}ds =\kappa \nu e^{-(\mu+\delta+\sigma)t}\int_0 ^t e^{-(\nu-(\delta +\sigma))s}ds.\] 
Thus, 
\begin{align*}\displaystyle c(t)= \left\{ \begin{array}{ll}\displaystyle \frac{\kappa \nu}{\nu-(\delta+\sigma)}\left(e^{-(\mu+\delta+\sigma)t}-e^{-(\nu+\mu)t}\right) \textrm{, if } \nu\neq \delta+\sigma,\\\\
 \displaystyle\kappa \nu te^{-(\mu+\delta+\sigma)t}\textrm{, if } \nu=\delta+\sigma.
 \end{array} 
 \right.
\end{align*}
Inserting this into Equation (\ref{c}), one gets
\begin{align*}\displaystyle \alpha= \left\{ \begin{array}{ll}\displaystyle -\left(\mu+\frac{\nu+\delta+\sigma}{2}\right)+\sqrt{\frac{(\nu-\delta-\sigma)^2}{4}+\kappa \nu} \textrm{, if } \nu \neq \delta+\sigma,\\\\
 \displaystyle\sqrt{\kappa\nu}-(\mu+\delta+\sigma) \textrm{, if } \nu=\delta+\sigma.
 \end{array} 
 \right.
\end{align*}
Then, 
\[\alpha=\displaystyle -\left(\mu+\frac{\nu+\delta+\sigma}{2}\right)+\sqrt{\frac{(\nu-\delta-\sigma)^2}{4}+\kappa \nu}.\]
\end{proof}
The basic reproduction number $R_0$, is the expected \textit{number} of secondary cases per primary case in a virgin population~\cite[page 4]{Odo}.
\begin{theorem}  The basic reproduction number of the epidemic is 
\begin{equation}
R_0=\frac{\nu\kappa}{(\mu+\nu)(\delta+\mu+\sigma)}.\label{R0}
\end{equation}
\end{theorem}
\begin{proof}
Let $Y$ be the \textit{number} of infectious contacts that  an individual has during  the infectious period. Then,
\[P(Y=0)=\frac{\mu}{\mu+\nu}+\frac{\nu}{\nu+\mu}\times\frac{\delta+\mu+\sigma}{\kappa+\delta+\mu+\sigma}.\]
The first term is the probability that the individual dies during the latent period, and the second term the probability that the individual does not die during the latent period but leaves the infectious compartment by death or recovery without an infectious contact. And for all positive integer $k$,
\[ P(Y=k)=\frac{\nu}{\mu+\nu}\times\left(\frac{\kappa}{\kappa+\delta+\mu+\sigma}\right)^k\times\frac{\delta+\mu+\sigma}{\kappa+\delta+\mu+\sigma}.\]
Therefore, $Z$ has a \emph{zero modified} geometric distribution~\cite[page 16]{Haccou}, with parameter $p=(\delta+\mu+\sigma)/(\kappa+\delta+\mu+\sigma)$ . Then the expected value of $Y$ is
\[E(Y)=\frac{\nu}{\mu+\nu}\times\frac{1-p}{p}=\frac{\nu}{\mu+\nu}\times\frac{\kappa}{\mu+\delta+\sigma}.\]
Thus, \[R_0=\frac{\nu}{\mu+\nu}\times\frac{\kappa}{\mu+\delta+\sigma}.\]
\end{proof}
\begin{remark}
The first factor is the probability that the individual does not die during the latent period, and the second is  the expected \textit{number} of infectious contacts while being infectious. An alternative way to derive $R_0$ is by the relation $R_0=\int_0 ^\infty c(t)dt$ \cite[page 69]{Haccou}. This formula gives the same result as above.
\end{remark}
As mentioned above, the SI, SIS, SIR, SIRS, SEI, SEIS, SEIR models are all sub-models of the SEIRS model.  Then, from the results obtained for the SEIRS model, we deduce that of the others. In Table \ref{thresholds}, we have the values of the basic reproduction number $R_0$ and the Malthusian parameter $\alpha$  of the models listed above. The fourth column gives the changes to get the corresponding model. Further by setting $\sigma=0$, one gets the corresponding results for a disease without an additional death rate for the infectious.
\begin{table}[h]
\begin{tabular}{|c|c|c|c|}
\hline
Model& $R_0 $&$\alpha$& Parameters change\\ \hline
$SEIRS$& $\frac{\kappa \nu}{(\mu+\nu)(\mu+\delta+\sigma)}$&$-\left(\mu+\frac{\nu+\delta+\sigma}{2}\right)+\sqrt{\frac{(\nu-\delta-\sigma)^2}{4}+\kappa \nu}$&\\ \hline
$SEIR$& $\frac{\kappa \nu}{(\mu+\nu)(\mu+\delta+\sigma)}$&$-\left(\mu+\frac{\nu+\delta+\sigma}{2}\right)+\sqrt{\frac{(\nu-\delta-\sigma)^2}{4}+\kappa \nu}$&$\rho=0$\\ \hline
$SEIS$& $\frac{\kappa \nu}{(\mu+\nu)(\mu+\delta+\sigma)}$&$-\left(\mu+\frac{\nu+\delta+\sigma}{2}\right)+\sqrt{\frac{(\nu-\delta-\sigma)^2}{4}+\kappa \nu}$&$\rho\rightarrow \infty$  \\ \hline
$SEI$& $\frac{\kappa \nu}{(\mu+\nu)(\mu+\sigma)}$&$-\left(\mu+\frac{\nu+\sigma}{2}\right)+\sqrt{\frac{(\nu-\sigma)^2}{4}+\kappa \nu}$&$\delta=0$, $\rho=0$\\ \hline
$SIRS$& $\frac{\kappa }{\mu+\delta+\sigma}$&$\kappa-(\mu+\delta+\sigma)$&$\nu \rightarrow\infty$\\ \hline
$SIR$ & $\frac{\kappa }{\mu+\delta+\sigma}$&$\kappa-(\mu+\delta+\sigma)$&$\nu \rightarrow\infty$, $\rho=0$\\ \hline
$SIS$& $\frac{\kappa }{\mu+\delta+\sigma}$&$\kappa-(\mu+\delta+\sigma)$&$\nu\rightarrow \infty$, $\rho\rightarrow \infty$  \\ \hline
$SI$& $\frac{\kappa }{\mu+\sigma}$&$\kappa-(\mu+\sigma)$&$\nu \rightarrow\infty$, $\delta=0$, $\rho=0$\\ \hline
\end{tabular}
 \caption{Thresholds for sub-models of the SEIRS model}\label{thresholds}  
\end{table}
\begin{remark}
The basic reproduction number $R_0$ and the Malthusian parameter $\alpha$ are identical for the SEIR, SEIS and SEIRS models.
\end{remark}

From  Equations (\ref{alpha}) and (\ref{R0}), we have
\begin{align*}
\alpha>0&\Longleftrightarrow-\left(\mu+\frac{\nu+\delta+\sigma}{2}\right)+\sqrt{\frac{(\nu-\delta-\sigma)^2}{4}+\kappa \nu}>0 \\
 &\Longleftrightarrow \frac{(\nu-\delta-\sigma)^2}{4}+\kappa \nu>{\left(\mu+\frac{\nu+\delta+\sigma}{2}\right)}^2\\
 &\Longleftrightarrow 4\kappa \nu>(2\mu+\nu+\delta+\sigma)^2-(\nu-\delta-\sigma)^2\\
 &\Longleftrightarrow 4\kappa \nu>(2\mu+2\nu)(2\mu+2\delta+2\sigma)\\
 &\Longleftrightarrow \frac{\kappa \nu}{(\mu+\nu)(\mu+\delta+\sigma)}>1\\
 &\Longleftrightarrow R_0>1.
 \end{align*}
 It follows that the basic reproduction number $R_0$ exceeds $1$ if and only if, the Malthusian parameter $\alpha$ exceeds $0$. That is, the sign relation $sign(\alpha) = sign(R_0-1)$ is verified.
\begin{remark} It is well known that to surely prevent the disease to invade the population, $R_0$ must be less than $1$. To control the epidemic one need then to diminish $R_0$.  The basic reproduction number $R_0$ can be written in the following form.
\[R_0=\kappa\times\frac{\nu}{\mu+\nu}\times\frac{1}{\mu+\delta+\sigma}.\]
So, it is clear that $R_0$ increases with $\kappa$ and $\nu$, and decreases with $\mu, \delta$ and $\sigma$. The contact rate  $\kappa$ can be reduced by hospitalization or quarantine of the infectious individuals. The recovery rate $\delta$ can be increased  by medication. In the case of an epizootic, $\sigma$ the disease related death rate, can be increased by culling infectious animals.
\end{remark}

We have derived two thresholds (the Malthusian parameter $\alpha$ and the basic reproduction number $R_0$) of the SEIRS epidemic branching process, deduced the corresponding thresholds for the sub-models of the SEIRS epidemic model and established that $sign(\alpha)=sign(R_0-1)$. 

\subsection{Main result for the initial phase}
In this subsection, we derive the probability for a minor outbreak of the epidemic branching process, and state the main result for the initial phase of the epidemic.

Let $\pi=P(\lim_t L_\infty(t)=0)$ be the probability of a minor outbreak of $L_\infty$ and $Y$ be the \textit{number} of infectious contacts that  an individual has during  the infectious period. As we start with one latent individual, $\pi$ is the smallest positive solution of the equation $z=g(z)$, where $g$ is the probability generating function (pgf) of $Y$ \cite[page 113]{Haccou}. We have 
\begin{align*}
g(z)&=\sum_{k=0} ^\infty P(Y=k)z^k\\
&=\frac{\mu}{\mu+\nu}+\frac{\nu}{\mu+\nu}\frac{\delta+\mu+\sigma}{\kappa+\delta+\mu+\sigma}+\sum_{k=1} ^\infty \frac{\nu}{\mu+\nu}\frac{\delta+\mu+\sigma}{\kappa+\delta+\mu+\sigma}{\left(\frac{\kappa}{\kappa+\delta+\mu+\sigma}\right)}^kz^k\\
&=a+\frac{(1-a)b}{1-(1-b)z},\textrm{ with } a=\frac{\mu}{\mu+\nu} \textrm{ and } b=\frac{\delta+\mu+\sigma}{\kappa+\delta+\mu+\sigma}.
\end{align*}
Then, $\pi$ is the smallest solution in $[0,1]$ of the following equation.
\begin{equation}
z=a+\frac{(1-a)b}{1-(1-b)z}.\label{extinction}
\end{equation}
Equation (\ref{extinction}) has two solutions, 
\[z_0=1 \textrm{ and } z_1=a+\frac{b}{1-b}=\frac{\mu}{\mu+\nu}+\frac{\delta+\mu+\sigma}{\kappa}=\frac{\mu}{\mu+\nu}+\frac{\nu}{\nu+\mu}\frac{1}{R_0}.\]
Then, we have the following theorem.
\begin{theorem}
Let $\pi$ be the probability of a minor outbreak of the epidemic  when started with one latent individual. Then,
\begin{eqnarray}\displaystyle \pi = \left\{ \begin{array}{ll}\displaystyle 1 & \textrm{ if } R_0 \leq 1,\\\\
 \displaystyle\frac{\mu}{\mu+\nu}+\frac{\nu}{\mu+\nu}\frac{1}{R_0}& \textrm{ if } R_0>1.
 \end{array} 
 \right.\label{pi}\end{eqnarray}
\end{theorem}
\begin{Corollary}
Let $\pi_{(m,k)}$ be the probability of a minor outbreak when the epidemic started with $m$ latent individuals and $k$ infectious individuals. Then, 
\begin{eqnarray}\displaystyle \pi_{(m,k)} = \left\{ \begin{array}{ll}\displaystyle 1 & \textrm{ if } R_0 \leq 1,\\\\
 \displaystyle\left(\frac{\mu}{\mu+\nu}+\frac{\nu}{\mu+\nu}\frac{1}{R_0}\right)^m\left(\frac{1}{R_0}\right)^k& \textrm{ if } R_0>1.
 \end{array} 
 \right.\label{pimk}\end{eqnarray}
\end{Corollary}
\begin{proof}
Let $\pi_{(m,k)}$ be the probability of a minor outbreak when the epidemic starts with $m$ latent and $k$  infectious individuals. With this notation, we  have\\
$\pi_{(1,0)}=P(\textrm{minor outbreak}|E(0)=1,I(0)=0)=\pi$, and\\
$\pi_{(0,1)}=P(\textrm{minor outbreak}|E(0)=0,I(0)=1)$. Thus, 
$\pi=\mu/(\mu+\nu)+(\nu/(\mu+\nu))\pi_{(0,1)}$.\\
Therefore, $\pi_{(0,1)}=((\mu+\nu)/\nu)\pi-\mu/\nu$.
Thus, by Equation (\ref{pi}) one gets
\begin{eqnarray}\displaystyle \pi_{(0,1)} = \left\{ \begin{array}{ll}\displaystyle 1 & \textrm{ if } R_0 \leq 1,\\\\
 \displaystyle\frac{1}{R_0}& \textrm{ if } R_0>1.
 \end{array} 
 \right.\label{pi01}
\end{eqnarray}
We have $\pi_{(m,k)}=\left(\pi_{(1,0)}\right)^m\left(\pi_{(0,1)}\right)^k$, since all the $m+k$ independent epidemics must die out \cite[page 112]{Haccou}. Thus, by Equations (\ref{pi}) and (\ref{pi01}), we get Equation (\ref{pimk}).
\end{proof}
\begin{remark}
This result is the same as that of the SEIR stochastic model studied by Allen and Lahodny in \cite{Allen}.
\end{remark}
As noted above, due to the additive death rate in the infectious compartment, if the epidemic takes off, the population process can be turned to a sub-critical process. In this case, the population may go extinct. As we start the process with $n$ individuals, assuming that $n$ is large,  we define the probability of minor outbreak of the epidemic, as the probability that the \textit{number} of infected individuals does not exceed $\sqrt{n}$, that is $P(L_n(t)<\sqrt{n}, \forall t\geq0)$ \cite[page 55]{Odo}. Let us state now the main result for the initial phase. 
\begin{theorem}\label{initial phase}
Consider the uniform SEIRS epidemic model defined above, with \\$(S_n(0),E_n(0), I_n(0),R_n(0))=(n-1,1,0,0), L_n(t)=E_n(t)+I_n(t), N_n(t)=S_n(t)+E_n(t)+I_n(t)+R_n(t)$, and let $L_\infty(t)$ denote the birth and death process defined above, $\alpha$ its  Malthusian parameter, $\pi$ the probability of a minor outbreak of $L_\infty$, and $\pi_n:=P(L_n(t)<\sqrt{n}, \forall t\geq0)$ denote the probability of a minor outbreak of the epidemic. Then as $n\rightarrow \infty$, we have the following results:
\begin{enumerate}
\item If $\alpha\leq0$, then for any $n, L_n(t)\rightarrow0$ as $t\rightarrow\infty$ with probability $1$.
\item If $0<\alpha<\lambda-\mu$, then  $\pi_n\rightarrow\pi=\mu/(\mu+\nu)+\nu/[(\mu+\nu)R_0]$. With the remaining probability $(1-\pi_n)\rightarrow1-\pi$, $L_n$ grows exponentially: $L_n(t)\sim e^{\alpha t}$, but $L_n(t)/N_n(t)\rightarrow0$ as $t\rightarrow\infty$.
\item If $\alpha>\lambda-\mu$ , then $\pi_n\rightarrow\pi=\mu/(\mu+\nu)+\nu/[(\mu+\nu)R_0]$. With the remaining probability $(1-\pi_n)\rightarrow1-\pi$, during the initial phase of the epidemic, $L_n$ grows exponentially with rate $\alpha$.
\end{enumerate}
\end{theorem}
\begin{proof}
\begin{enumerate}
\item
If $\alpha<0$, then $L_\infty$ is sub-critical and dies out with probability $1$. If $\alpha=0$, then $L_\infty$ is critical and dies out with probability $1$, since $P(Y=1)\neq1$ \cite{Haccou}.
\item
If $\alpha>0$, then $L_\infty$ is super-critical. It dies out with probability $\pi$ and with the remaining probability $1-\pi$, it grows exponentially with rate $\alpha$, that is $L_\infty\sim e^{\alpha t}$. As $N_n(t)\sim e^{(\lambda-\mu)t}$, if $\alpha<\lambda-\mu$, then $L_n(t)/N_n(t)\longrightarrow0$, when $t\longrightarrow\infty$.
\item 
If $\alpha>\lambda-\mu$, then $\alpha>0$ since $\lambda-\mu>0$. Thus, $L_\infty$ is super-critical. It dies out with probability $\pi$ and with the remaining probability $1-\pi$, it grows exponentially with rate $\alpha$. As $L_n \Longrightarrow L_\infty$, the same applies to $L_n$.
\end{enumerate}
\end{proof}
\begin{remark}
If $\alpha>\lambda-\mu$, then $L_\infty$ is super-critical. Thus the epidemic may take off. If it does, the \textit{number} of infected individuals grows initially with the rate $\alpha$ that is larger than the initial growth rate of the population ($\lambda-\mu$). After the initial phase, in the case of a major outbreak, several scenarios are possible. i)The population goes on growing exponentially, eventually with a lower rate; ii) due to the additive death rate of the infectious, the effective death rate of the population becomes larger than its birth rate, and the population process becomes a sub-critical process. The different scenarios are treated in Section \ref{deterministic} and illustrated by simulations in Section \ref{Simulations}.
\end{remark}
We have derived the probability of a minor outbreak of the epidemic branching process, stated and shown the main result of the initial phase of the epidemic.

\section{The deterministic SEIRS model}\label{deterministic}
Now we consider the corresponding deterministic model of the stochastic model studied above. As the population size is varying, we study  first the fractions system and then deduce the asymptotic behavior of the compartments sizes. In this section the deterministic sizes of the compartments S, E, I, R and the population size at the time $t$ are denoted  $S(t), E(t), I(t), R(t)$ and $N(t)$ respectively.
 
The corresponding deterministic SEIRS model of the model above is given by the following system of ordinary differential equations (ODE).
\begin{equation}
\begin{split}
&\frac{dS}{dt}=\lambda N+\rho R-\kappa S\frac{I}{N}-\mu S,\\
&\frac{dE}{dt}=\kappa S\frac{I}{N}-(\nu+\mu)E,\\
&\frac{dI}{dt}=\nu E-(\delta+\mu+\sigma)I,\\
&\frac{dR}{dt}=\delta I-(\rho+\mu)R,\\
&N=S+E+I+R,\\
&S(0)>0, E(0)>0, I(0)>0, R(0)\geq0.
\end{split}\label{SEIRSsystem}
\end{equation}
\begin{remark}
System (\ref{SEIRSsystem}) is the same as System (2) studied by Greenhalgh in \cite{Greenhalgh}, with a constant contact rate ($\beta(N)=\beta$), a constant death rate ($f(N)=\mu$) and without vaccination ($p=q=0$). But Greenhalgh assumed that the death rate $f(N)$ is a strictly monotone increasing continuously differentiable function of N. So, the model that we study is not a sub-model of that of Greenhalgh since we consider a constant death rate.
\end{remark}
From System (\ref{SEIRSsystem}), we have $dN/dt=(\lambda-\mu)N-\sigma I=(\lambda-\mu-\sigma i)N$, where $i$ is the \textit{fraction} of the infectious individuals. Thus, the population should grow if $\lambda>\mu+\sigma i$, stabilize if $\lambda=\mu+\sigma i$ and decrease if $\lambda<\mu+\sigma i$.
\begin{theorem}
 $N(t)$ is constant and positive ($N(t)=N(0)>0, \forall t>0$), if and only if the parameters verify the following equality:
\begin{equation}
\begin{split}
(\rho+\mu)\nu\kappa\sigma\lambda +\rho\kappa\nu\delta(\lambda-\mu)& \\
-(\rho+\mu)(\nu+\mu)(\delta+\mu+\sigma)[\kappa(\lambda-\mu)+\mu\sigma]&=0,
\end{split}\label{equation1}
\end{equation}
and the initial values verify
\begin{equation} \displaystyle \left\{ \begin{array}{l}
S(0)=(\nu\kappa)^{-1}(\nu+\mu)(\delta +\mu +\sigma)N(0),\\\\
E(0)=(\nu\sigma)^{-1}(\delta +\mu +\sigma)(\lambda-\mu)N(0),\\\\
I(0)=\sigma^{-1}(\lambda-\mu)N(0),\\\\
R(0)=(\sigma(\rho+\mu))^{-1}\delta(\lambda-\mu)N(0),\\\\
\textrm{ with } N(0)>0.
\end{array} \right.\label{equation12}
\end{equation}\label{SEIRSconstant}
\end{theorem}
\begin{proof}
By using successively the derivatives of $N, I, E, R$ and $S$, one gets that $N(t)$ is constant and positive, if and only if $I, E, R$ and $S$ are constant, the parameters verify Equation (\ref{equation1}) and the initial values verify System (\ref{equation12}).
\end{proof}
Generally, Equation (\ref{equation1}) and System (\ref{equation12}) are not verified, thus $N(t)$ is not constant. Therefore, we consider the \textit{fractions} $s=S/N, e=E/N, i=I/N$ and $r=R/N$. By System (\ref{SEIRSsystem}), one gets
\begin{equation}
\begin{split}
&\frac{ds}{dt}=\lambda-\lambda s+\rho r+(\sigma-\kappa)si,\\
&\frac{de}{dt}=\kappa si-(\lambda+\nu)e+\sigma ei,\\
&\frac{di}{dt}=\nu e-(\lambda+\delta+\sigma)i+\sigma i^2,\\
&\frac{dr}{dt}=-(\lambda+\rho)r+\delta i +\sigma ri,\\
&s+e+i+r=1.
\end{split}
\label{seirsf system}
\end{equation}
This system is equivalent to System (2.2) in \cite{Greenhalgh}, with $p=q=0$.
\begin{remark}
The natural death rate $\mu$ does not intervene in the derivatives of the fractions. This is logic, since this rate  is the same for all the compartments, it has no effect on the fractions.
\end{remark}
Since $r=1-s-e-i$, it is enough to study the system
\begin{equation}
\begin{split}
&\frac{ds}{dt}=\lambda+\rho-(\lambda+\rho) s-\rho e-\rho i+(\sigma-\kappa)si,\\
&\frac{de}{dt}=\kappa si-(\lambda+\nu)e+\sigma ei,\\
&\frac{di}{dt}=\nu e-(\lambda+\delta+\sigma)i+\sigma i^2,
\end{split}\label{seirs system}
\end{equation}
in the domain 
\begin{equation}
D=\{(s,e,i);s\geq0, e\geq0, i\geq0, s+e+i\leq 1\}.
\end{equation}
\begin{theorem}
The domain D is positively invariant for System \ref{seirs system}.
\end{theorem}
\begin{proof}
If $s=0$, then $ds/dt=\lambda+\rho(1-e-i)>0$. If $e=0$, then $de/dt=\kappa si\geq0$.
If $i=0$, then $di/dt=\nu e\geq0$. If $s+e+i=1$, then $d(s+e+i)/dt=-\delta i\leq0$.
Then every solution of System \ref{seirs system} starting in $D$, remains there for all $t>0$. That is $D$ is positively invariant for System \ref{seirs system}
\end{proof}
In the following, we use the next generation matrix (NGM) described in \cite{Odo1} to derive a threshold parameter, with threshold value $1$ for System (\ref{seirs system}).\\
By setting $ds/dt=de/dt=di/dt=0$ with $e=i=0$ in System (\ref{seirs system}), we get $s=1$. Then, $(1,0,0)$ is the unique disease free equilibrium (DFE) of System (\ref{seirs system}). $e$ and $i$ are the infected \textit{fractions} of the model. Thus, the \emph{infected subsystem} is 
\begin{equation*}
\begin{split}
&\frac{de}{dt}=\kappa si-(\lambda+\nu)e+\sigma ei,\\
&\frac{di}{dt}=\nu e-(\lambda+\delta+\sigma)i+\sigma i^2.
\end{split}
\end{equation*}
Hence, the \emph{linearized infected subsystem} at the DFE is 
\begin{equation}
\begin{split}
&\frac{de}{dt}=\kappa i-(\lambda+\nu)e,\\
&\frac{di}{dt}=\nu e-(\lambda+\delta+\sigma)i.
\end{split}\label{linearinfectedeq}
\end{equation}
Let $x=(e,i)^t$ be the vector of the infected \textit{fractions}. Thus, System (\ref{linearinfectedeq}) is equivalent to
\[\stackrel{.}{x}=(T+\Sigma)x,\]
with 
\[T=\left(\begin{array}{cc} 0&\kappa\\0&0\end{array}\right) \textrm{ and } \Sigma=\left(\begin{array}{cc}-(\lambda+\nu)&0\\ \nu&-(\lambda+\delta+\sigma) \end{array}\right).\]
$T$ is the transmissions matrix and $\Sigma$ is the transitions matrix. Then, the next generation matrix with large domain is 
\[K_L = -T\Sigma^{-1}.\]
Let $R_1$ be the spectral radius of $K_L$. We have 
\begin{equation}
 R_1=\rho(K_L)=\frac{\kappa\nu}{(\lambda+\nu)(\lambda+\delta+\sigma)}.\label{R1}
\end{equation} 
 An equilibrium is said to be stable if nearby solutions stay nearby for all future time \cite[p. 175]{Hirsch}. More precisely an equilibrium $x^*$ is said to be stable, if for every neighborhood $V$ of $x^*$ there is a neighborhood $V_1$ of $x^*$, such that every solution starting in $V_1$ remains in $V$ for all $t>0$. If $V_1$ can be chosen such that $\lim_{t\rightarrow\infty} x(t)=x^*$, then $x^*$ is said to be asymptotically stable. An equilibrium is said to be unstable, when it is not stable. An equilibrium $x^*$ is said to be globally asymptotically stable (GAS) in an invariant set $D$, $(x^*\in D)$, if it is locally stable and $\lim_{t\rightarrow\infty} x(t)=x^*$, for every solution $x(t)$ starting in $D$.  
\begin{theorem}
The disease free equilibrium (DFE) of System (\ref{seirs system}) is globally asymptotically stable (GAS) in $D$, if $R_1\leq1$, and unstable if $R_1>1$.  \label{DFE}
\end{theorem}
\begin{proof}
Let $f(s,e,i)$ be the Rhs of System (\ref{seirs system}). Then, 
\[f(s,e,i)=\left(\begin{array}{c}\lambda+\rho-(\lambda+\rho)s-\rho e-\rho i+(\sigma-\kappa)si\\
\kappa si-(\lambda+\nu)e+\sigma ei\\
\nu e-(\lambda+\delta+\sigma)i+\sigma i^2\end{array}\right).\]
The Jacobian of $f$ at the disease free equilibrium is 
\[Df(1,0,0)=\left(\begin{array}{ccc}-(\lambda+\rho)&-\rho&\sigma-\kappa-\rho\\0&-(\lambda+\nu)&\kappa\\0&\nu&-(\lambda+\delta+\sigma)
\end{array}\right).\]
The characteristic polynomial of $Df(1,0,0)$ is 
\[P(x)=(-\lambda-\rho-x)[x^2+(2\lambda+\nu+\delta+\sigma)x+(\lambda+\nu)(\lambda+\delta+\sigma)-\nu\kappa].\]
$-(\lambda+\rho)$ is an evident negative root of $P(x)$. Thus, by the Routh-Hurwitz criterion \cite[page 11]{Routh},  all the roots of $P(x)$ has negative real part if and only if $(\lambda+\nu)(\lambda+\delta+\sigma)-\nu\kappa>0$. And we have $(\lambda+\nu)(\lambda+\delta+\sigma)-\nu\kappa>0\Longleftrightarrow R_1<1$.
Thus, the disease free equilibrium is locally asymptotically stable if $R_1<1$, and unstable if $R_1>1$.

Let $V$ denote the function defined on $D$ by $V(s,e,i)=\nu e+(\lambda+\nu)i$. Then, 
\begin{align*}\stackrel{.}{V}(s,e,i)&=\nu\frac{de}{dt}+(\lambda+\nu)\frac{di}{dt}\\
&=\nu[\kappa si-(\lambda+\nu)e+\sigma ei]+(\lambda+\nu)[\nu e-(\lambda+\delta+\sigma)i+\sigma i^2]\\
&=i[\nu\kappa s+\nu\sigma e+(\lambda+\nu)\sigma i-(\lambda+\nu)(\lambda+\delta+\sigma)]\\
&=iL(s,e,i),
\end{align*}
with $L(s,e,i)=\nu\kappa s+\nu\sigma e+(\lambda+\nu)\sigma i-(\lambda+\nu)(\lambda+\delta+\sigma)$. The affinity of $L$ implies that it achieves its maximum at the extreme points of the boundary of the closed set $D$. But $L(0,0,0)=-(\lambda+\nu)(\lambda+\delta+\sigma), L(0,0,1)=-(\lambda+\nu)(\lambda+\delta), L(0,1,0)=-\lambda\sigma-(\lambda+\nu)(\lambda+\delta)$ and $L(1,0,0)=\nu\kappa-(\lambda+\nu)(\lambda+\delta+\sigma)=(\lambda+\nu)(\lambda+\delta+\sigma)(R_1-1)$. Thus, $\stackrel{.}{V}\leq0$ in $D$ if $R_1\leq1$. Then, $V$ is a Lyapunov function of System (\ref{seirs system}). The only invariant subset of the set with $\stackrel{.}{V}=0$ is $\{(1,0,0)\}$. It follows from LaSalle's Invariance Principle \cite[p. 200]{Hirsch}, that  the disease free equilibrium (DFE) is globally asymptotically stable (GAS) in $D$, when $R_1\leq1$. 
\end{proof}
By Equations (\ref{R1}) and (\ref{alpha}), we have
\begin{align*}
R_1>1&\Longleftrightarrow\frac{\kappa\nu}{(\lambda+\nu)(\lambda+\delta+\sigma)}>1\\
 &\Longleftrightarrow\kappa\nu>(\lambda+\nu)(\lambda+\delta+\sigma)\\
 &\Longleftrightarrow\kappa\nu>\frac{1}{4}\left((2\lambda+\nu+\delta+\sigma)^2-(\nu-\delta-\sigma)^2\right)\\
 &\Longleftrightarrow\frac{(\nu-\delta-\sigma)^2}{4}+\kappa\nu>\frac{(2\lambda+\nu+\delta+\sigma)^2}{4}\\
 &\Longleftrightarrow\sqrt{\frac{(\nu-\delta-\sigma)^2}{4}+\kappa\nu}>\lambda+\frac{\nu+\delta+\sigma}{2}\\
 &\Longleftrightarrow-\left(\mu+\frac{\nu+\delta+\sigma}{2}\right)+\sqrt{\frac{(\nu-\delta-\sigma)^2}{4}+\kappa\nu}>\lambda-\mu\\
 &\Longleftrightarrow\alpha>\lambda-\mu.
\end{align*}
It follows that the fraction's threshold $R_1$ exceeds $1$ if and only if the Malthusian parameter  $\alpha$, exceeds the initial growth rate of the population $\lambda-\mu$. That is, we have the sign relation $sign(R_1-1)=sign(\alpha-(\lambda-\mu))$. Thus, the global stability of the disease free equilibrium of the fraction's system when $R_1<1$, confirms that if $\alpha<\lambda-\mu$, then the infected \textit{fraction} vanishes even if the epidemic takes off (Theorem (\ref{initial phase})  (ii)).

Greenhalgh has shown \cite[Theorem 2.3]{Greenhalgh} that if $R_1>1$, then System (\ref{seirs system}) has at least one endemic equilibrium, and that this equilibrium is unique and locally asymptotically stable (LAS) when the average duration of immunity exceeds both the average infectious and incubation periods, that is $\delta>\rho$ and $\nu>\rho$. We have not proved, but we strongly believe that if $R_1>1$, then System  (\ref{seirs system}) has one and only one endemic equilibrium, and that this equilibrium is globally asymptotically stable in the interior of $D$. The simulations that we made support this conjecture (Figure \ref{seirs1} (c) and (d)).
\begin{theorem}
Let $(S(t),E(t),I(t),R(t))$ be a solution of System (\ref{SEIRSsystem}) and $R_0$ denoted the basic reproduction number given by Equation (\ref{R0}).
\begin{enumerate}
\item If $R_0<1$, then $(S(t),E(t),I(t),R(t))\longrightarrow(\infty,0,0,0)$;
\item if $R_0=1$, then $(S(t),E(t),I(t),R(t))\longrightarrow(\infty,E^*,I^*,R^*)$,\\ with $E^*>0,I^*>0,R^*>0$;
\item if $R_0>1\geq R_1$, then $(S(t),E(t),I(t),R(t))\longrightarrow(\infty,\infty,\infty,\infty)$.
\end{enumerate}\label{asympt1}
\end{theorem}
\begin{remark}
The case $R_1>1$ is treated in Theorem \ref{R2}.
\end{remark}
\begin{proof} We have \[R_0=\frac{\nu\kappa}{(\mu+\nu)(\mu+\delta+\sigma)} \textrm{ and } R_1=\frac{\nu\kappa}{(\lambda+\nu)(\lambda+\delta+\sigma)}.\]
Then, $R_0>R_1$, since $\lambda>\mu$. Therefore, in the three cases of the Theorem \ref{asympt1}, we have $R_1\leq1$. Let us assume that $R_1\leq1$. Thus, by Theorem (\ref{DFE}), $(s,e,i,r)\longrightarrow(1,0,0,0)$ when $t\longrightarrow\infty$. $dN/dt=(\lambda-\mu)N-\sigma I=(\lambda-\mu-\sigma i)N$. Then, $dN/dt\longrightarrow(\lambda-\mu)N$, when $t\longrightarrow\infty$.
Thus, $N\longrightarrow\infty$, when $t\longrightarrow\infty$  because $\lambda>\mu$.
Therefore, $S\longrightarrow\infty$, when  $t\longrightarrow\infty$, since  $S/N\longrightarrow1$. By using the derivatives of $E$ and $I$, one gets
\[\left(\frac{E}{I}\right)'=\kappa s+(\delta+\sigma-\nu)\frac{E}{I}-\nu\left(\frac{E}{I}\right)^2. \textrm{ Where the prime denotes the derivative}.\]
\[\textrm{Then, }\left(\frac{E}{I}\right)'\longrightarrow\kappa +(\delta+\sigma-\nu)\frac{E}{I}-\nu\left(\frac{E}{I}\right)^2 \textrm{, when } t\longrightarrow\infty.\]
Thus $E/I$ can be approximate by a solution of the following equation, when $t\longrightarrow\infty$.
\begin{equation}
y'=\kappa+(\delta+\sigma-\nu)y-\nu y^2 \label{Ric}
\end{equation} 
Equation (\ref{Ric}) is a Riccati's equation \cite{Haaneim}. By solving it, one gets
\begin{equation*}
y:t\longmapsto \left(Ce^{\sqrt{\Delta}t}-\frac{\nu}{\sqrt{\Delta}}\right)^{-1}+\frac{\delta+\sigma-\nu+\sqrt{\Delta}}{2\nu}, \textrm{ with } C > 0 \textrm{, where } \Delta=(\delta+\sigma-\nu)^2+4\nu\kappa. 
\end{equation*} 
Then, $E/I\longrightarrow (\delta+\sigma-\nu+\sqrt{\Delta})/(2\nu)$, when  $t\longrightarrow\infty$.\\
 We have $dI/dt=\nu E-(\delta+\mu+\sigma)I=[\nu(E/I)-(\delta+\mu+\sigma)]I$. Thus, by substituting $E/I$ by its asymptotic value, one gets 
 \[\frac{dI}{dt}\longrightarrow\left[\frac{\delta+\sigma-\nu+\sqrt{\Delta}}{2}-(\delta+\mu+\sigma)\right]I \textrm{, when } t\longrightarrow\infty;\]
and 
\begin{align*}
\frac{\delta+\sigma-\nu+\sqrt{\Delta}}{2}-(\delta+\mu+\sigma)&=-\left(\mu+\frac{\delta+\sigma+\nu}{2}\right)+\frac{\sqrt{\Delta}}{2}\\
&=-\left(\mu+\frac{\delta+\sigma+\nu}{2}\right)+\sqrt{\frac{(\delta+\sigma-\nu)^2}{4}+\kappa\nu}\\
&=\alpha \textrm{, the Malthusian parameter given by Equation (\ref{alpha})}.
\end{align*}
Therefore, $dI/dt\longrightarrow\alpha I$, when $t \longrightarrow\infty$.\\ 
As $dE/dt=\kappa SI/N-(\nu+\mu)E$, $S/N\longrightarrow1$ and $E/I\longrightarrow(\delta+\sigma-\nu+\sqrt{\Delta})/(2\nu)$, after some algebra, one gets $dE/dt\longrightarrow\alpha E$, when $t\longrightarrow\infty$.
As $sign(\alpha)=sign(R_0-1)$, 
\begin{align*}
(E,I)\longrightarrow&(0,0) \textrm{ if } R_0<1;\\
(E,I)\longrightarrow&(E^*,I^*) \textrm{, with } E^*>0\textrm{ and }I^*>0, \textrm{ if } R_0=1;\\
(E,I)\longrightarrow&(\infty,\infty) \textrm{ if } R_0>1.
\end{align*}
For the \textit{number} of the recovered $R$, as $dR/dt=\delta I-(\rho+\mu)R$, it is obvious that $R$ has the same asymptotic behavior as $I$.
\end{proof}
\begin{remark}
In the proof, we have shown that if $R_1\leq1$, then the Malthusian parameter $\alpha$ of the stochastic model, is also the common asymptotic growth rate of  the compartments E and I, and $\lambda-\mu$ is the asymptotic growth rate of the population. Since $sign(R_1-1)=sign(\alpha-(\lambda-\mu))$, this is coherent with Theorem \ref{initial phase} (ii).
\end{remark}
Theorem \ref{asympt1} gives the asymptotic behavior of the compartments sizes, when $R_1\leq1$. If $R_1>1$, then the fraction disease free equilibrium is unstable. Therefore, the disease will remain endemic in the population in term of the \textit{fraction} infected. The following theorem gives the asymptotic behavior of the compartments sizes, when $R_1>1$, assuming that the fraction system admits an endemic equilibrium that is globally asymptotically stable in the interior of the feasible region $D$. 
\begin{theorem}
Assume that  $R_1>1$ and that System (\ref{seirs system}) has a unique endemic equilibrium $(s^*, e^*, i^*)$ that is globally asymptotically stable in $\stackrel{o}{D}$, and set 
\begin{equation}
R_2=\frac{\lambda}{\mu+\sigma i^*}.
\end{equation}
\begin{enumerate}
\item If $R_2>1$, then $(S,E,I,R)\longrightarrow(\infty,\infty,\infty,\infty)$.
\item If $R_2=1$, then $(S,E,I,R)\longrightarrow(S^*,E^*,I^*,R^*)$, \\
with $S^*>0, E^*>0, I^*>0$, $R^*>0$.
\item If $R_2<1$, then $(S,E,I,R)\longrightarrow(0,0,0,0)$.
\end{enumerate}\label{R2}
\end{theorem}
\begin{proof}
Let us assume that there is an endemic equilibrium $(s^*,e^*,i^*)$ of the fraction System (\ref{seirs system}) and that it is globally asymptotically stable in $\stackrel{o}{D}$. Then,
\[\frac{dN}{dt}\longrightarrow(\lambda-\mu-\sigma i^*)N \textrm{, when }t\longrightarrow\infty.\]
Asymptotically, the population would increase with rate $\lambda$, and decrease with rate $\mu+\sigma i^*$. As the fraction system admits an endemic equilibrium, that is globally asymptotically stable in the interior of the feasible region, all the compartments have the same asymptotic behavior as the population. Let us set $\alpha_2=\lambda-\mu-\sigma i^*$. The quantity $\alpha_2$ is the common asymptotic exponential growth/decay rate of all the compartments S, E, I and R.  We have $sign(\alpha_2)=sign(R_2-1)$. Thus, the results follow.
\end{proof} 
In this section we have studied the corresponding deterministic SEIRS model of the previous stochastic model. We derived a threshold quantity $R_1$ for the fraction model. If $R_1\leq1$, then the \textit{fraction}'s disease free equilibrium is globally asymptotically stable in the feasible region $D$, otherwise it is unstable. When $R_1\leq1$, the behavior of the \textit{number} of infected is determined by the basic reproduction number $R_0$. If $R_0<1$, then the \textit{number} of infected vanishes. If $R_0=1$, then the \textit{number} of infected stabilizes to a positive value; when $R_1\leq1<R_0$, then the \textit{number} of infected grows exponentially, but at a lower rate than the population.  If $R_1>1$, then the \textit{number} of infected grows initially quicker than the population and the asymptotic behavior of the population is governed by the threshold quantity $R_2$. If $R_2<1$, then the population vanishes; if $R_2=1$, then the population stabilizes; if $R_2>1$, then the population grows, but with a lower rate than its initial growth rate.
\section{Simulations}\label{Simulations}
In this section, we use the software R  to illustrate and confirm the results found in the previous sections. In the following, we set $\mu=1$, that is the time unit is the life expectancy, except for the simulations of influenza epidemics in Burkina Faso where we set one year as the time unit. The other parameters and the initial values are chosen arbitrary, unless otherwise stated. 
\subsection{Simulations of the initial phase}
In this subsection, we give some examples of simulations of  epidemics starting by one latent individual, and using different values of the parameters. The population is initiated with $1$ latent individual and $999$ susceptible individuals. 
 
In Figure \ref{SEIRSsimR0leq1} (a) and (b), where $R_0=0.41$ and $R_0=0.73$  respectively, all the $10$ epidemics die without any major outbreak. In (a) the maximum of infected individuals is $2$, while it is $6$ in (b).
\begin{figure}[!h]
\centering
\subfigure[$R_0=0.414$]{
\includegraphics[scale=0.3]{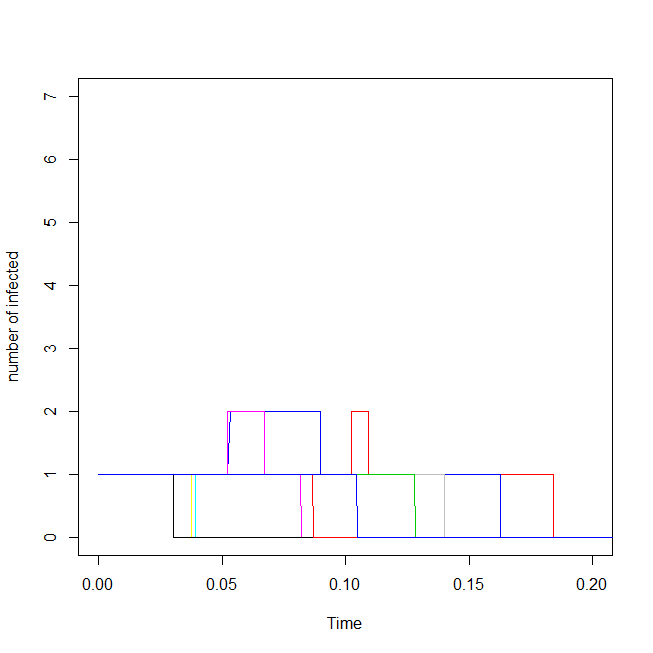}
}
\quad
\subfigure[$R_0=0.733$]{
\includegraphics[scale=0.3]{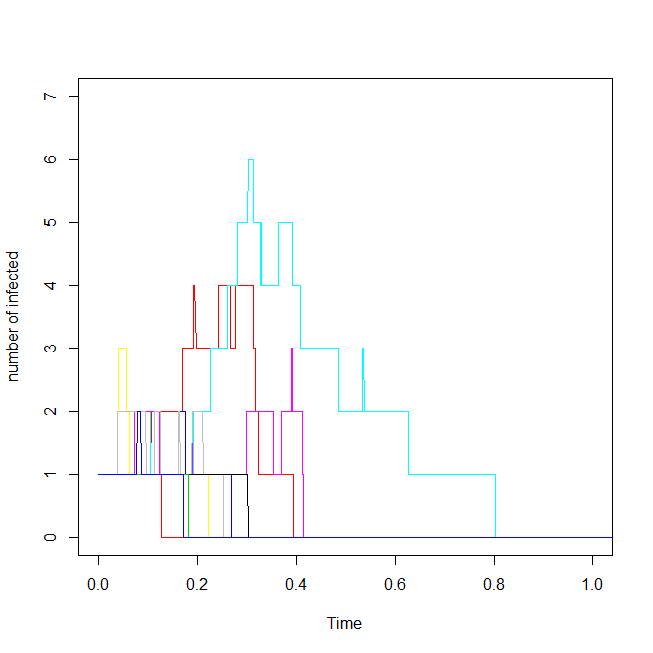}
}
\caption{ (a)  10 SEIRS simulations with $\lambda=3, \mu=1, \sigma=7, \delta=15, \kappa=10, \nu=20, \rho=5$, that gives$ R_0=0.41, \alpha=-7.82, \pi=1$. (b) 10 SEIRS simulations with  $\lambda=1.2, \mu=1, \sigma=7, \delta=5, \kappa=10, \nu=20, \rho=5$, that gives $R_0=0.73, \alpha=-2.30, \pi=1$. In both cases, all the $10$ epidemics die out without a major outbreak. However, they die out quicker and the number of infected is fewer in (a) than in (b).}\label{SEIRSsimR0leq1}
\end{figure}

In Figure \ref{SEIRSsimR0g1}, where $R_0=2$, four simulated epidemics out of $10$ have a major outbreak. The other $6$ epidemics die out without many getting infected. For the epidemics with major outbreak, the \textit{number} of the infected individuals grow exponentially but the time where the exponential growth starts varies. 
\begin{figure}[!h]
\centering
\subfigure[]{
\includegraphics[scale=0.3]{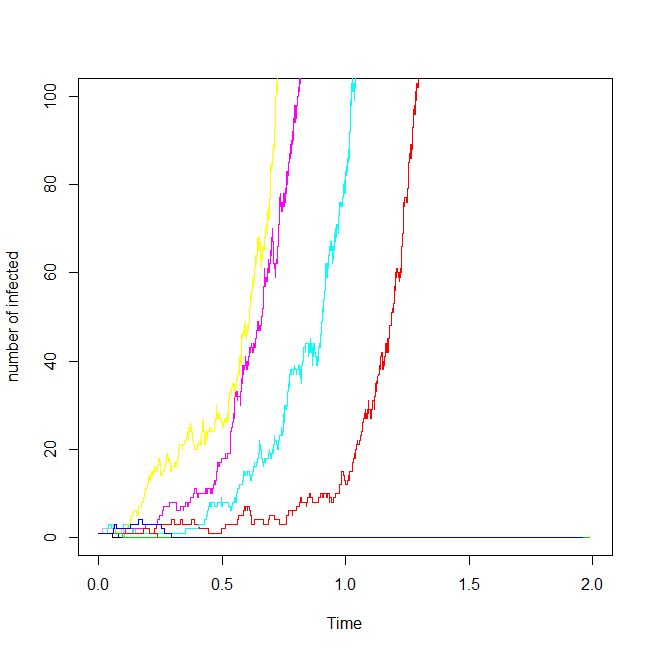}
}
\quad
\subfigure[]{
\includegraphics[scale=0.3]{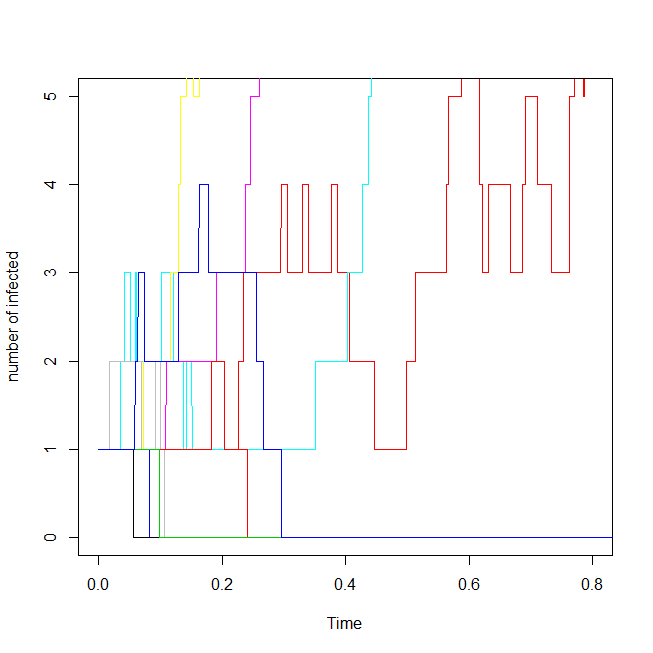}
}
\caption{10 SEIRS simulations ($6$ dying quickly)  with  $\lambda=3, \mu=1, \sigma=4, \delta=5, \kappa=21, \nu=20, \rho=5, n=1000, R_0=2, R_1=1.52, \alpha=5.72, \pi=0.52 $. In (b) we made a zoom so that the six minor epidemics can be seen.}\label{SEIRSsimR0g1}
\end{figure}

Now we estimate the probability of a minor outbreak $\pi$ by simulating $1000$ epidemics and setting $\hat{\pi}_n=n_0/1000$, where $n_0$ is the number of minor epidemics. We set $\lambda=3, \mu=1, \nu=50, \delta=10, \sigma=4, \rho=3$, and set successively $\kappa=10, 20, 30, 50, 100$ to get different values of $\pi$. Table \ref {minout} gives the different values of  $\pi$ and  the estimate $\hat{\pi}_n$ for $n=1000$ and $n=2000$ respectively. These results confirm that the probability of extinction of the branching process $L_{\infty}$ is a good approximation of the probability of a minor outbreak of the epidemic starting with one latent individual in a population of size $n$, when $n$ is large.
\begin{table}[h]
\begin{centering}
\begin{tabular}{|l|l|l|l|l|l|l|l|l|l|l|}
\hline
$n$&$1000$&$2000$&$1000$&$2000$&$1000$&$2000$&1000&$2000$&$1000$&$2000$\\ \hline
$\kappa$&$10$&$10$&$20$&$20$&$30$&$30$&$50$&$50$&$100$&$100$\\ \hline
$R_0$&$0.654$&$0.654$&$1.307$&$1.307$&$1.961$&$1.961$&$3.268$&$3.268$&$6.536$&$6.536$\\ \hline
$\pi$&1.000&$1.000$&$0.770$&$0.770$&$0.520$&$0.520$&$0.320$&$0.320$&$0.170$&$0.170$\\ \hline
$\hat{\pi}_n$&$1.000$&$1.000$&$0.788$&$0.782$&$0.513$&$0.520$&$0.303$&$0.329$&$0.165$&$0.172$\\ \hline
\end{tabular}
\caption{Estimation of the probability of minor outbreak $\pi_n$ of the epidemic starting with one latent individual and $n-1$ susceptible individuals. The theoretical result is $\pi$ and the simulated result is $\hat{\pi}_n$ from $1000$ simulations.}\label{minout}
\end{centering}
\end{table}
  
These simulations confirm that when $R_0$ is less than $1$, the disease cannot invade the population, and that if $R_0$ is larger than $1$, then with a positive probability $(1-\pi_n)\longrightarrow(1-\pi)$, the disease can invade the population, and in this case the \textit{number} of the infected individuals grows exponentially during the initial phase. 
\subsection{ Simulations of major outbreaks}
In this subsection, we show some  simulations of major outbreaks, where the epidemic starts with a positive \textit{fraction} of infected individuals. So the simulations illustrate what happens once the \textit{number} of infected has reached a small but positive \textit{fraction} of the community. We use the blue color for the susceptible, green for the exposed, red for the infectious and black for the recovered. 

We start by simulating the deterministic \textit{fraction}'s system. In Figure \ref{seirs1}, we have four cases with $R_1=0.5, 1, 1.76, 5.33$ respectively. In each case we have ten solutions paths of System (\ref{seirsf system}) with different initial values. In (a) as in  (b), the $10$ solutions of System (\ref{seirsf system}) approach the disease free equilibrium, confirming that if $R_1\leq1$, then the disease free equilibrium of the deterministic fraction system is globally asymptotically stable in the feasible region. In (c), all the $10$ solutions  approach the same endemic equilibrium $(s^*, e^*, i^*, r^*)\approx(0.50,0.13,0.14, 0.24)$. In (d), all the $10$ solutions  approach the same endemic equilibrium $(s^*, e^*, i^*, r^*)\approx(0.12, 0.49, 0.31, 0.08)$. The results of (c) and (d) confirm that when $R_1>1$, the disease free equilibrium is unstable and that there is an endemic equilibrium that is globally asymptotically stable in the interior of the feasible region. That is the endemic level is independent of the starting value. However the endemic equilibrium of (d) is different of that of (c), hence the endemic equilibrium vary with the parameters values.
\begin{figure}[!h]
\centering
\subfigure[$R_1=0.5$]{
\includegraphics[scale=0.3]{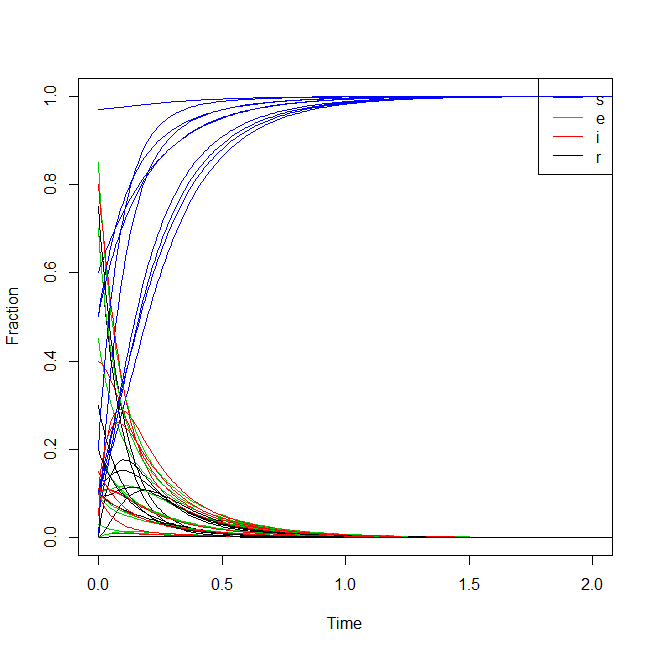}
}
\quad
\subfigure[$R_1=1$]{
\includegraphics[scale=0.3]{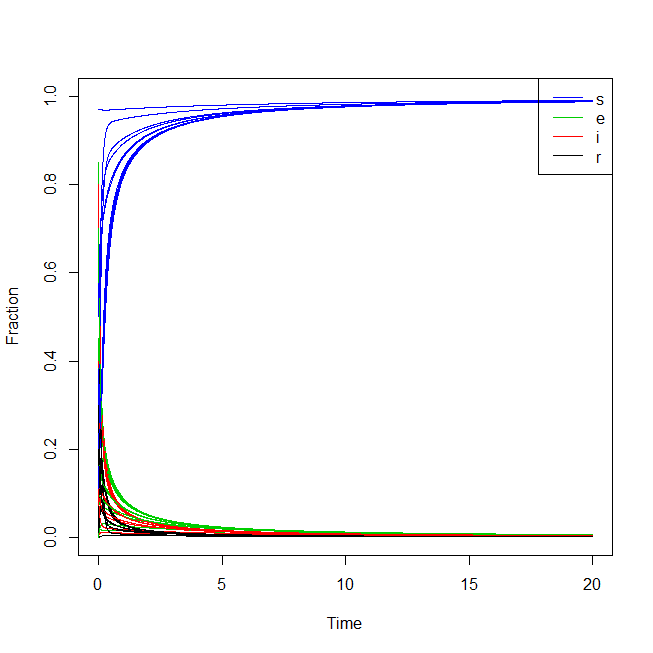}
}
\quad
\subfigure[$R_1=1.76$]{
\includegraphics[scale=0.3]{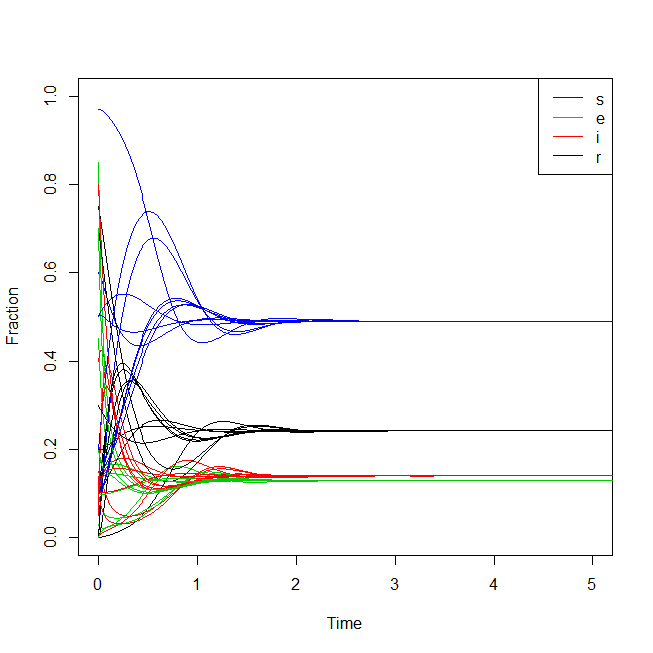}
}
\quad
\subfigure[$R_1=5.33$]{
\includegraphics[scale=0.3]{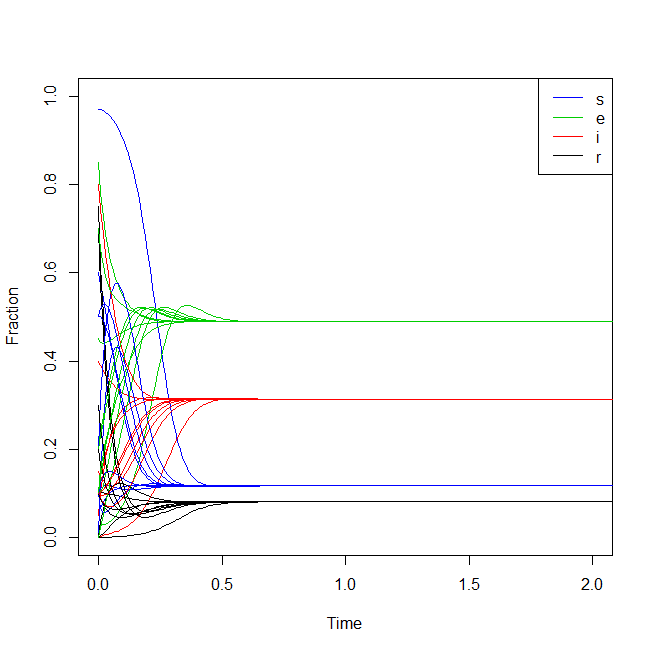}
}
\caption{Simulations of of System (\ref{seirsf system}). In each case we have 10 solutions paths of System (\ref{seirsf system}) with different initial values. For (a)  we have $\lambda=2, \mu=1, \sigma=8, \delta=5, \kappa=9, \nu=10, \rho=10$ that gives $R_1=0.5$. For (b) the parameters have the same values as in (a) except that we set $\kappa=18$ to get $R_1=1$. In (a) and in (b) all the solutions approach the disease free equilibrium $(1,0,0,0)$. The time scale is longer in (b), so the epidemic takes longer time to die out when $R_1$ is close to $1$.
 For (c) we have  $\lambda=2, \mu=1, \sigma=8, \delta=5, \kappa=30, \nu=15, \rho=2$, that gives $ R_1=1.76$; all the solutions approach the same endemic equilibrium $(s^*, e^*, i^*, r^*)\approx(0.50,0.13,0.14, 0.24)$. For (d), we have $\lambda=2, \mu=1, \sigma=8, \delta=5, \kappa=100, \nu=8, \rho=20$ that gives $R_1=5.33$; all the solutions approach the same endemic equilibrium $(s^*, e^*, i^*, r^*)\approx(0.12, 0.49, 0.31, 0.08)$.} \label{seirs1}
\end{figure}

In the following we simulate both the stochastic and the deterministic models. In each case, we have simulated the stochastic epidemic, as well as integrated numerically the  deterministic system (\ref{SEIRSsystem}), both starting at the same values. From the \textit{numbers}, we got the \textit{fractions}  by setting $s=S/N, e=E/N, i=I/N, r=R/N, \textrm{ with }N=S+E+I+R$. One distinguishes the stochastic solutions from the deterministic by the fact that the deterministic solutions are represented by smooth lines, while the solutions of the stochastic solutions are represented by broken lines.

In Figure \ref{SEIRSsd4}, where $R_0=1.27>1$ and $R_1=0.95<1$, the \textit{numbers} of latent, infectious and recovered grow exponentially as the population, but the population growth rate is even larger and the \textit{fractions} of the infected compartments go to \textit{zero}. This means that, in term of the \textit{number} of infected individuals, the disease is endemic, but the disease dies out in term of the \textit{fractions}. 
\begin{figure}[!h]
\vspace{-0.3cm}
\centering
\subfigure[\textit{numbers}]{
\includegraphics[scale=0.3]{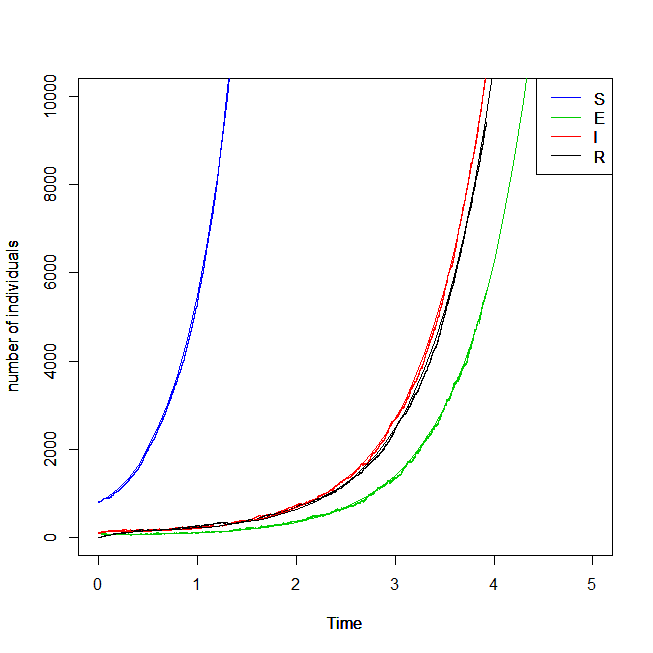}
}
\quad
\subfigure[\textit{Fractions}]{
\includegraphics[scale=0.3]{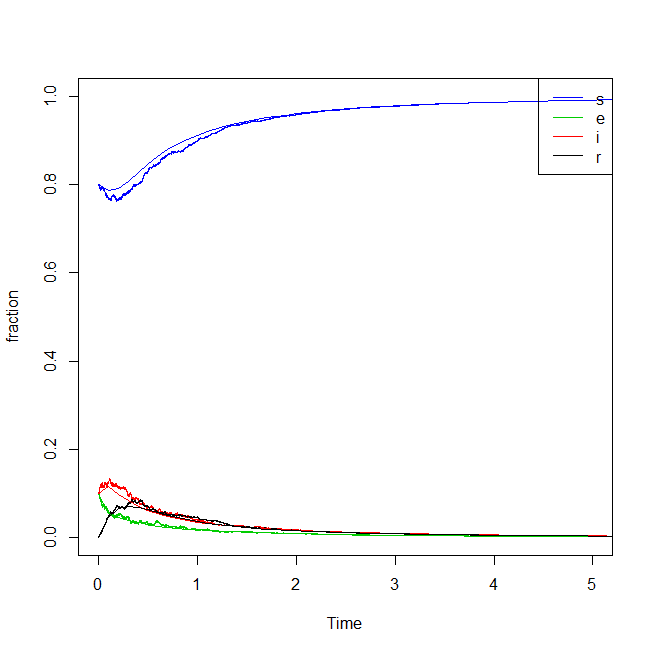}
}
\caption{SEIRS curves  with  $\lambda=3, \mu=1, \sigma=3, \delta=5, \kappa=12, \nu=20, \rho=3$ that gives $R_0=1.27, R_1=0.95, \pi=0.80, \alpha=1.61$ The initial values are $N(0)=1000$ with $(S(0),E(0),I(0),R(0))=(800,100,100,0)$. All the numbers grow exponentially. But the population grow faster, and the fractions go to the disease free equilibrium. One distinguishes hardly the deterministic curves from the stochastic because they are very close.} \label{SEIRSsd4}
\end{figure}

In Figure \ref{SEIRSsd1}, initially the \textit{number} of the infected grows quicker than the \textit{number} of the susceptible. After that, the \textit{number} of the susceptible declines. Afterward all the compartments grow exponentially, but with a rate $\alpha_2\approx0.16$  that is lower than the initial growth rate of the population $\lambda-\mu=1$.  The \textit{fractions} go to an endemic equilibrium. 
\begin{figure}[!h]
\centering
\subfigure[Numbers]{
\includegraphics[scale=0.3]{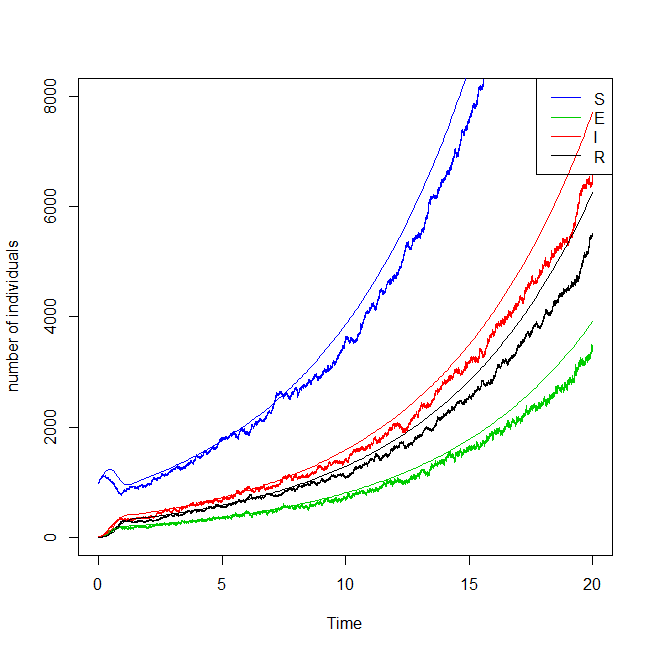}
}
\quad
\subfigure[\textit{Fractions}]{
\includegraphics[scale=0.3]{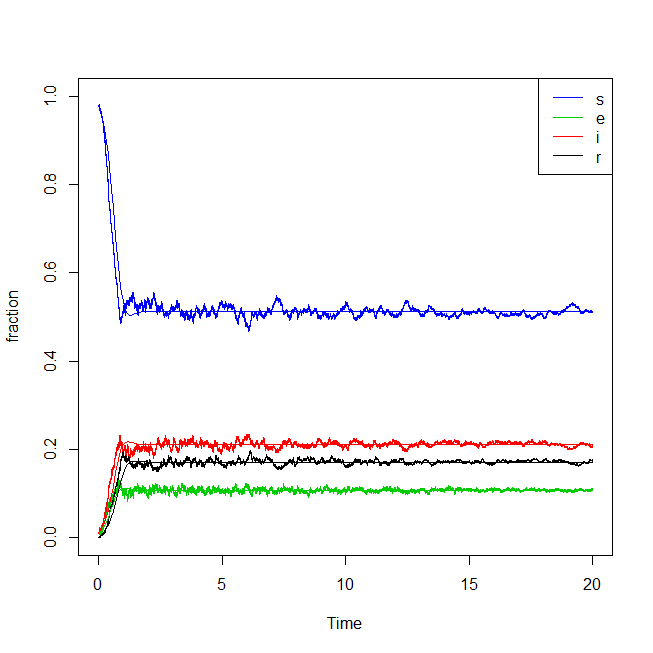}
}
\caption{SEIRS curves  with  $\lambda=2, \mu=1, \sigma=4, \delta=5, \kappa=21, \nu=20, \rho=5$ that gives $R_0=2, R_1=1.74, \pi=0.52, \alpha=5.72$. The initial values are  $(S(0),E(0),I(0),R(0))=(980,10,10,0)$. All the \emph{numbers} grow exponentially with rate $\alpha_2\approx0.16$, while the \emph{fractions} go to an endemic equilibrium $(s^*,e^*,i^*,r^*)\approx(0.51, 0.11, 0.21, 0.17)$.} \label{SEIRSsd1}
\end{figure}

For Figure \ref{SEIRSsd5}, the parameters are chosen such that Equation (\ref{equation1}) is verified, allowing then the existence of endemic equilibrium for the deterministic System \ref{SEIRSsystem}. The asymptotic reproduction \textit{number} of the population $R_2=1$, then for the  deterministic solution, the population stabilizes, when $t\longrightarrow\infty$. The stochastic \textit{numbers} fluctuate around the deterministic \textit{numbers}. The \textit{fractions} have the same behavior as that of the \textit{numbers}. 
\begin{figure}[!h]
\centering
\subfigure[\textit{Numbers}]{
\includegraphics[scale=0.3]{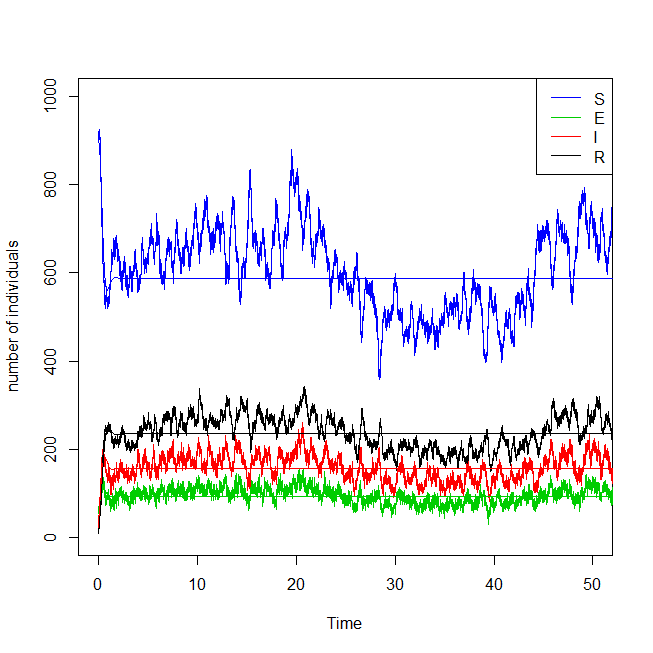}
}
\quad
\subfigure[\textit{Fractions}]{
\includegraphics[scale=0.3]{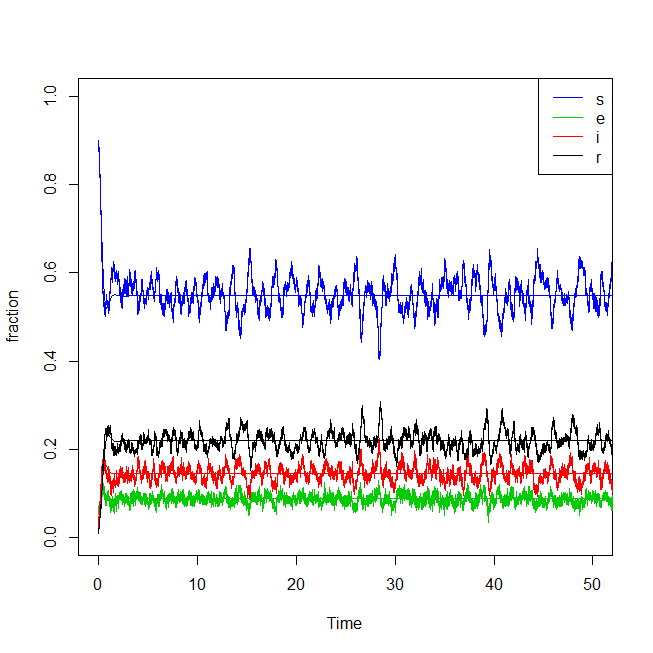}
}
\caption{SEIRS curves  with  $\lambda=1.73, \mu=1, \sigma=5, \delta=6, \kappa\approx 23.02, \nu=20, \rho=3$ that gives $R_0=1.83, R_1=1.66, \alpha=5.42$, The initial values are  $(S(0),E(0),I(0),R(0))=(900,70,20,10)$. The asymptotic growth rate of the population is $\alpha_2=0$, its asymptotic reproduction number $R_2=1$. Thus for the deterministic solutions the \emph{numbers} approach an endemic equilibrium $(S^*,E^*,I^*,R^*)\approx(581,  94, 157, 235)$ and the \emph{fractions} go to an endemic equilibrium $(s^*,e^*,i^*,r^*)\approx(0.55, 0.09, 0.15, 0.22)$. The stochastic solutions fluctuate around the deterministic.} \label{SEIRSsd5}
\end{figure}

In Figure \ref{SEIRSsd3}, for the deterministic model, the epidemic turned the population exponential growth to an exponential decay due to the disease induced death rate, while the \textit{fractions} go to an endemic equilibrium. For the stochastic model, the population size first decreases but at some instance when there are only few remaining individuals, the disease goes extinct and then  the population size starts growing again. The deterministic model suggests that the population will go extinct, whereas in the stochastic model the disease first dies out and then the population becomes super critical again, thus regrowing. In the stochastic setting, what happens when the numbers become low is random. Both the disease and the population could die out, or the population can grow again after the extinction of the epidemic. Only analyzing the deterministic \textit{fraction} model would give a misleading conclusion since the \textit{fractions} seem to stabilize, whereas in what really happens is that all \textit{numbers} in the deterministic model tend to $0$.
\begin{figure}[!h]
\centering
\subfigure[\textit{Numbers}]{
\includegraphics[scale=0.3]{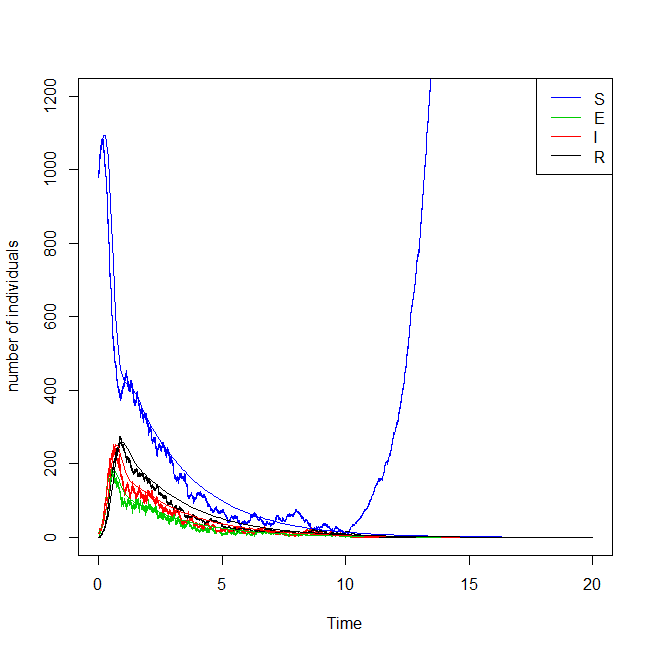}
}
\quad
\subfigure[\textit{Fractions}]{
\includegraphics[scale=0.3]{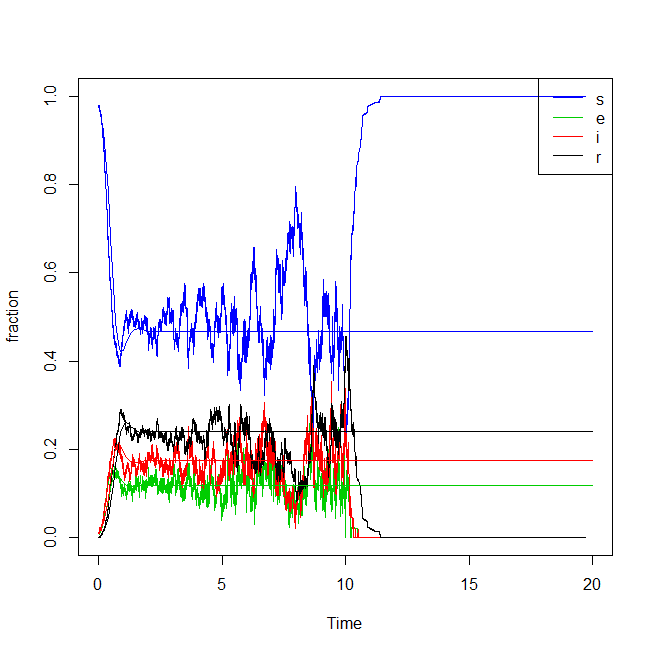}
}
\caption{SEIRS curves  with  $\lambda=2, \mu=1, \sigma=8, \delta=5, \kappa=30, \nu=20, \rho=3$ that gives $R_0=2.04, R_1=1.81, \pi=0.51, \alpha=7.24$, the initial values are $(S(0),E(0),I(0),R(0))=(980,10,10,0)$. The asymptotic decay rate of the population is $\alpha_2\approx-0.39$, its asymptotic reproduction number $R_2\approx0.84$. $R_1>1$ and $R_2<1$ thus, in the deterministic solutions, the \emph{numbers} vanish, while the \emph{fractions} go to an endemic equilibrium $(s^*,e^*,i^*,r^*)\approx(0.48, 0.12, 0.17, 0.24)$. In the stochastic solution, all the numbers vanish except the number of the susceptible that decreases until the extinction of the disease and regrow exponentially after that.} \label{SEIRSsd3}
\end{figure}

In Figure \ref{SEIRSsd6}, we have the case where $R_0=1$. For the stochastic model the disease goes extinct, while it persists in the deterministic one. In both models the population goes on growing exponentially. In (b) we made a zoom to see the dynamics of $E, I$ and $R$. 
\begin{figure}[!h]
\centering
\subfigure[\textit{Numbers $R_0=1$}]{
\includegraphics[scale=0.3]{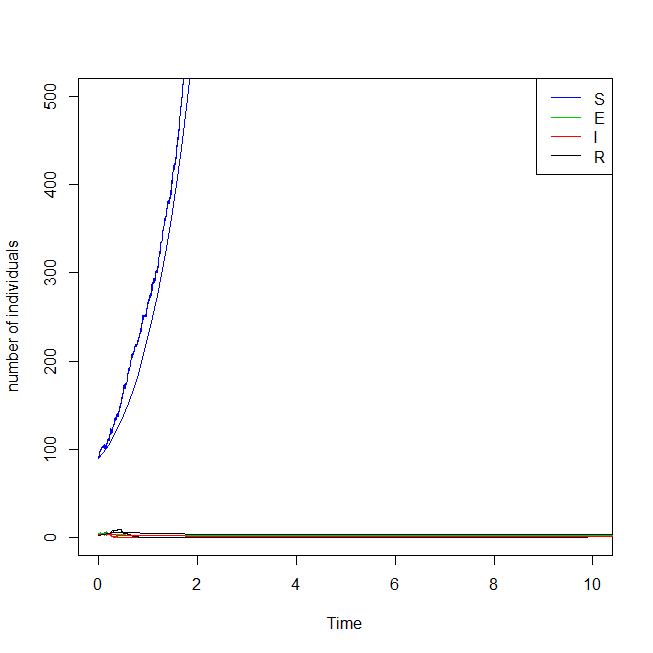}
}
\quad
\subfigure[\textit{Numbers $R_0=1$ (Zoom-in)}]{
\includegraphics[scale=0.3]{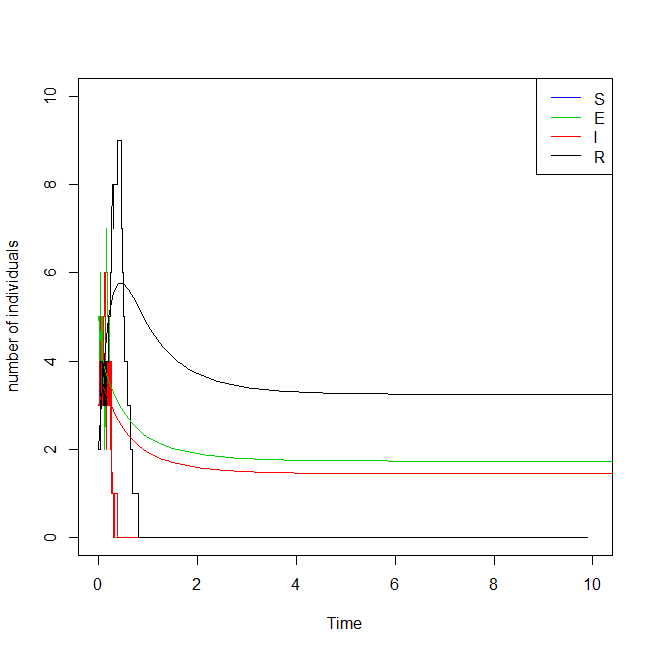}
}
\caption{SEIRS curves  with  $\lambda=2, \mu=1, \sigma=8, \delta=9, \kappa=19.2, \nu=15, \rho=3$ that gives $R_0=1$, the initial values are  $(S(0),E(0),I(0),R(0))=(90,5,3,2)$. The population grows exponentially in the stochastic and in the deterministic model. In the stochastic model the latent, the infectious and the recovered vanish, while in the deterministic model, they stabilize to positive values $(E^*, I^*, R^*)\approx(1.73, 1.44 ,3.24)$. In (b) the scale is chosen to show the dynamics of $E, I$ and $R$; $S(t)$ is much larger and not shown. } \label{SEIRSsd6}
\end{figure}

These simulations show the different possible scenarios in the case of a major outbreak. They confirm the theoretical results and show the similarities and differences between the stochastic model and the deterministic model. 
\subsection{Simulation of influenza epidemics in Burkina Faso}
Now we simulate two influenza epidemics in Burkina Faso with two different basic reproduction numbers.  Burkina Faso is an inland country of West Africa. Its population in $2016$ is about $16~000~000$. The individual annual birth rate and  death rate are estimated to $\lambda=0.046$ and $\mu=0.0118$ respectively \cite{Insd}. 

According to the World Health Organization (WHO)  \cite{WHO}, influenza is caused by a virus that attacks mainly the respiratory tract: the nose, the throat, the bronchi and rarely also the lungs. The infection usually lasts for about a week. It is characterized by sudden onset of high fever, myalgia, headache and severe malaise, non-productive cough, sore throat, and rhinitis. Most people recover within one to two weeks without requiring any medical treatment. The virus is easily passed from person to person through the air by droplets and small particles excreted when infected individuals cough or sneeze. The influenza virus enters the body through the nose or the throat. It then takes between one and four days for the person to develop symptoms. Someone suffering from influenza can be infectious from the day before he/she develops symptoms until seven days afterwards. The disease spreads very quickly among the population especially in crowded circumstances. Cold and dry weather enables the virus to survive longer outside the body than in other conditions and, as a consequence, seasonal epidemics in temperate areas appear in winter. Much less is known about the impact of influenza in the developing world. However, influenza outbreaks in the tropics where viral transmission normally continues year-round tend to have high attack and case-fatality rates. Therefore by setting one year as the time unit, we make the following estimation for the influenza parameters. The latent period is approximately $2.5$ days, the infectious period is approximately $7$ days and the immunity period is approximately one year, thus  $\nu=365/2.5$, $\delta=365/7$, $\rho=1$. We set the influenza case fatality rate (CFR) to $0.1\%$ and deduce the  influenza related death rate $\sigma\approx0.0522$. 

The reproduction number of the $1918$ pandemic influenza is estimated to be between $2$ and $3$ \cite{Mills}. Thus we set $R_0=2.5$ and deduce the contact number $\kappa=130$ from the other parameters and Equation (\ref{R0}). We simulate the epidemic for a period of $10$ years starting in $2016$ by integrating numerically the deterministic System (\ref{SEIRSsystem}). Figure \ref{Burkina} gives the evolution of the \textit{numbers} of susceptible, latent, infectious and recovered individuals and the corresponding fractions during the $10$ years. Figures \ref{Burkina} (c) and  (d), show respectively the dynamics of $E$ and $I$ and that of the \textit{fractions} $e$ and $i$. We have a peak with  $2~783~834$ infectious individuals at the $11^{th}$ week of the epidemic. After that the \textit{number} of infectious declines because the \textit{number} of susceptible is low. Afterward we have a minor peak every year due to the immunity waning and the newborns that increase the \textit{number} of susceptible. The \textit{number} of recovered individuals grow quickly and reach its maximum $13~139~592$ at the $17^{th}$ week. The \textit{fractions} go to an endemic equilibrium through damped oscillations. The population in $2026$ is estimated to $22~370~000$ individuals with $8~960~000$ susceptible, $94~000$ latent individuals, $260~000$ infectious individuals and $13~060~000$ recovered with non-permanent immunity.  The number of recovered individuals is larger than the \textit{number} of susceptible individuals. The \textit{fractions} go to an endemic equilibrium with more than $1\%$ of the population infected at every time. As the influenza last about one week and we have $52$ weeks within a year, more than $50\%$ of the population should be infected during the year $2026$. That will have a very important negative impact on the economy of the country. We have $R_1\approx2.50$ and $R_2\approx3.71$, thus $R_1$ and $R_2$ are both larger than $1$ and according to Theorem \ref{R2} the \emph{fractions} should go to an endemic equilibrium and all the compartments should grow exponentially. Therefore the simulations agree with the theoretical results. 
\begin{figure}[!h]
\centering
\subfigure[\textit{Numbers}]{
\includegraphics[scale=0.3]{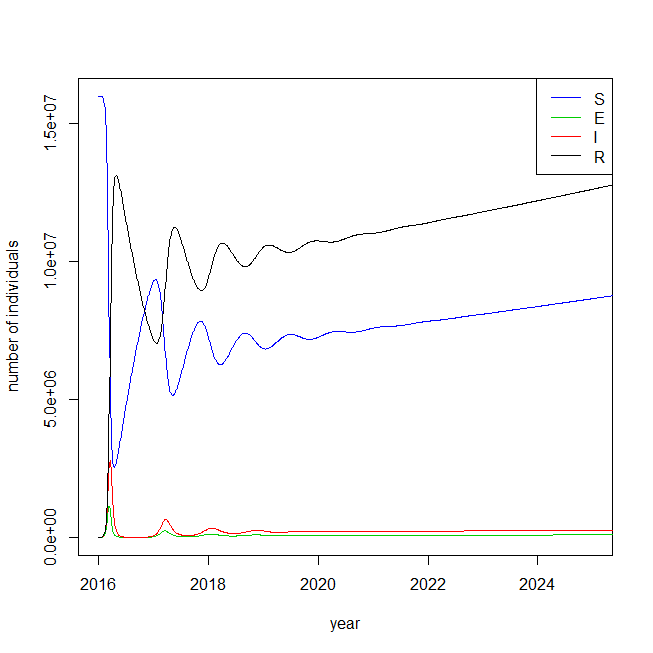}
}
\quad
\subfigure[\textit{Fractions}]{
\includegraphics[scale=0.3]{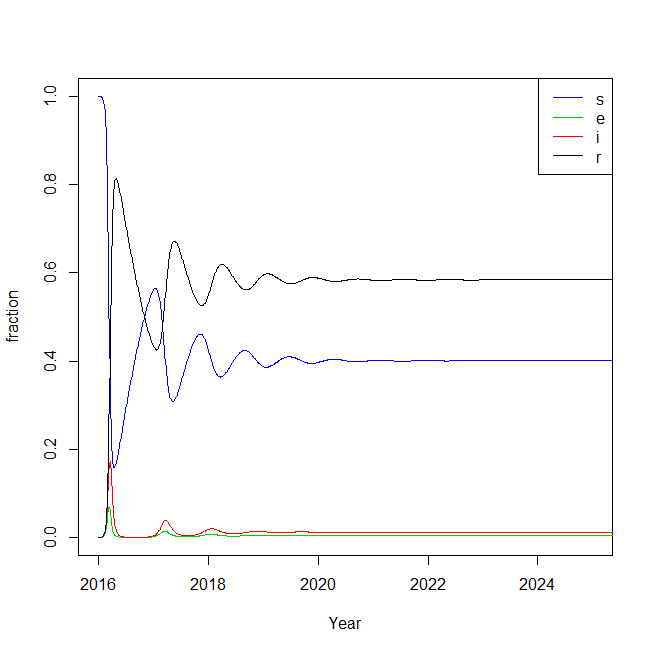}
}
\quad
\subfigure[\textit{Numbers of infected}]{
\includegraphics[scale=0.3]{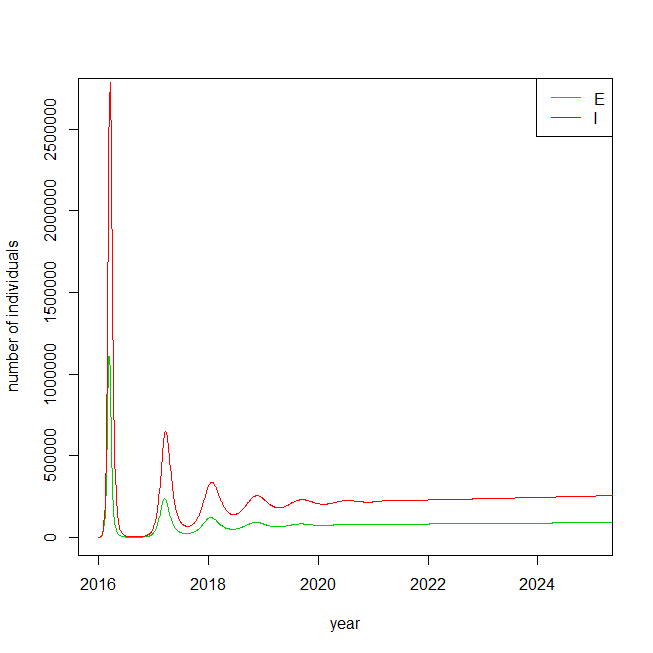}
}
\quad
\subfigure[\textit{Fractions of infected}]{
\includegraphics[scale=0.3]{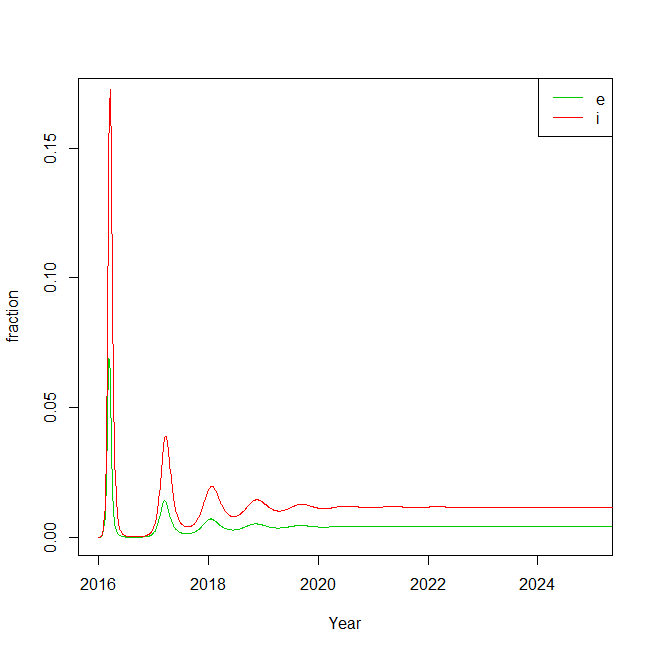}
}
\caption{Simulation of influenza in Burkina Faso with $R_0=2.5$. The parameters values are  $\lambda=0.046, \mu=0.0118, \sigma=0.0522, \delta=365/7, \kappa= 130.5277, \nu=365/2.5, \rho=1$ that gives $R_0=2.5$, $R_1=2.499$. The initial values are  $(S(0),E(0),I(0),R(0))=(16~000~000,1~000,400,10)$. By damped oscillations, the fractions approach an endemic equilibrium  $(s^*,e^*,i^*,r^*)\approx(0.400, 0.004, 0.012, 0.584)$. The initial growth rate of the population is $\lambda-\mu=0.0342$, with the epidemic its asymptotic growth rate  is $\alpha_2\approx0.0336$ and its asymptotic reproduction number rate is $R_2\approx3.707$. (c) show the dynamics of $E$ and $I$; (d) show the dynamics of $e$ and $i$.} \label{Burkina}
\end{figure}

The basic reproduction number for the novel influenza A (H1N1) has been estimated to be between $1.4$ and $1.6$ \cite{Coburn}. Thus, we set now $R_0=1.5$ and deduce the contact number from the other parameters. The results of this simulation are shown in  Figure \ref{Burkina1}. In this case the epidemic and the population have globally the same dynamics as above. But the impact of the epidemic is fewer. The major peak of the epidemic happens later, at the $25^{th}$ week, with $794~004$ infectious individuals. Contrary to the preceding case, the number of the recovered individuals is below the number of the susceptible individuals. 
\begin{figure}[!h]
\centering
\subfigure[\textit{Numbers}]{
\includegraphics[scale=0.3]{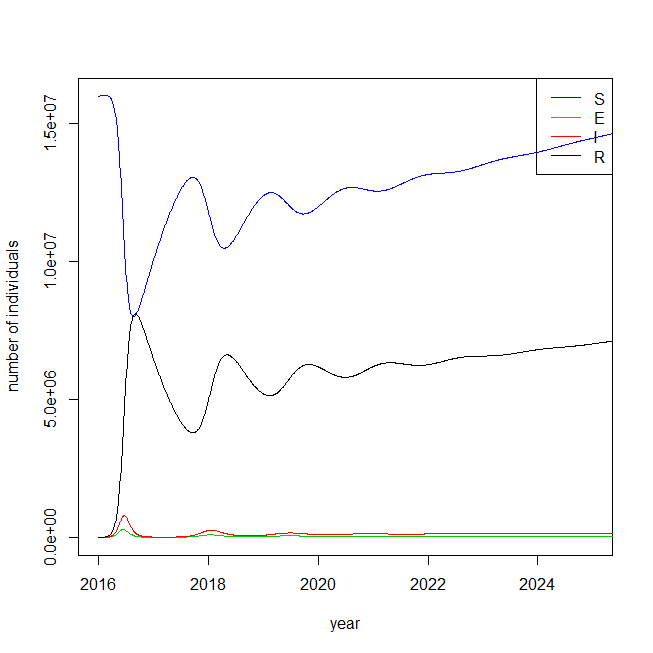}
}
\quad
\subfigure[\textit{Fractions}]{
\includegraphics[scale=0.3]{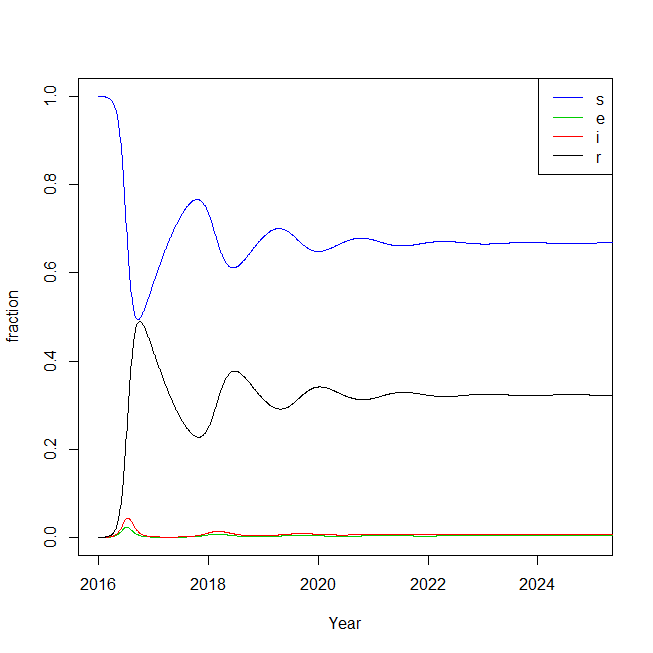}
}
\quad
\subfigure[\textit{Numbers}]{
\includegraphics[scale=0.3]{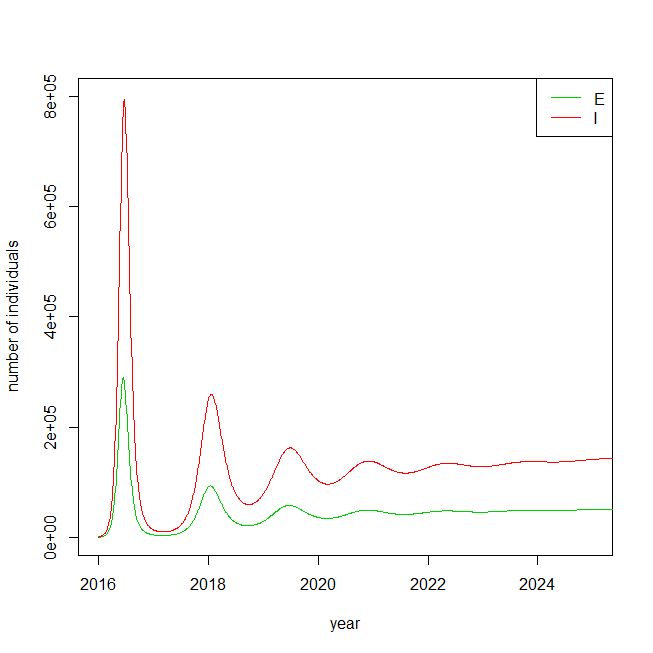}
}
\quad
\subfigure[\textit{Fractions}]{
\includegraphics[scale=0.3]{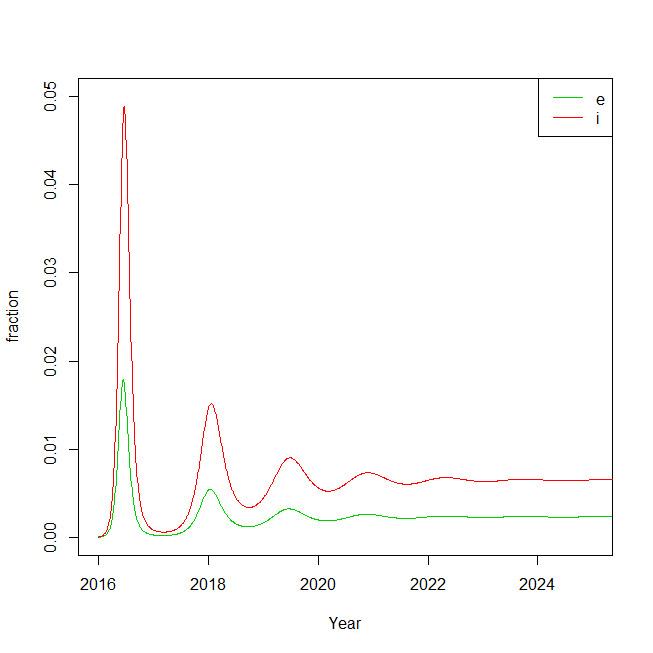}
}
\caption{Simulation of influenza in Burkina Faso with $R_0=1.5$. The parameters values are  $\lambda=0.046, \mu=0.0118, \sigma=0.0522, \delta=365/7, \kappa=78.3166, \nu=365/2.5, \rho=1$ that gives $R_0=1.5$, $R_1=1.499$. The initial values are  $(S(0),E(0),I(0),R(0))=(16~000~000,1~000,400,10)$. By damped oscillations, the fractions approach an endemic equilibrium  $(s^*,e^*,i^*,r^*)\approx(0.667, 0.002, 0.006, 0.324)$. The initial growth rate of the population is $\lambda-\mu=0.0342$, with the epidemic its asymptotic growth rate  is $\alpha_2\approx0.0339$ and its asymptotic reproduction number is $R_2\approx3.79$. (c) show the dynamics of $E$ and $I$; (d) show the dynamics of $e$ and $i$.} \label{Burkina1}
\end{figure}

In \cite{Coburn}, Coburn, Wagner and Blower simulated an influenza epidemic using a SIR model with demography. Their result for the \textit{number} of infectious individuals \cite[Figure 2 (a)]{Coburn} is similar to that of Figure \ref{Burkina} (c) and Figure \ref{Burkina1} (c). We assume no seasonal effects. Adding seasonality should make seasonal effects remain \cite[Figure 2 (b)]{Coburn}.
The simulations of influenza in Burkina Faso show that in spite of the epidemic the population should go on growing. But the number of infected individuals will grow also.  Furthermore the peak of the epidemic in the first year show that the emergence of a new strain of influenza virus will be a very serious threat for the world. The Global Influenza Program (GIP) of the World Health Organization (WHO), provides Member States with strategic guidance, technical support and coordination of activities essential to make their health systems better prepared against this threat.

In this section we have illustrated and validated the theoretical results of the previous sections by simulations. 

\section{Conclusion and discussions}\label{Conclusion}
In this paper we have studied a stochastic SEIRS epidemic model, with a disease related death, in a population which grows exponentially without the disease. We assumed that initially, the population process is a super-critical linear birth and death process with birth rate $\lambda$ and death rate $\mu$. 

We  have derived, the basic reproduction number $R_0$, the Malthusian parameter $\alpha$ and the probability of minor outbreak $\pi$ assuming that the initial size $n$ of the population tends to infinity. Considering the deterministic model, we derived the threshold quantity $R_1$ for the \textit{fractions}. 

If $R_0\leq1$, then the disease dies out with probability $1$. That is, there is no possibility of major outbreak if $R_0\leq1$. In this case, the population remains a super-critical process.

If $R_0>1$ then, with a positive probability, the epidemic can take off. If the epidemic takes off, then the \textit{number} of the infected (exposed or infectious) individuals  grows exponentially with rate $\alpha$. If $0<\alpha\leq\lambda-\mu$, then the \textit{sizes} of all the compartments grow exponentially while the \textit{fraction} of the infected individuals goes to  $0$. The \textit{number} of infected people grows, but at a lower rate than the population, implying that the \textit{fraction} infected becomes negligible. If $\alpha>\lambda-\mu$, then the \textit{number} of infected individuals grows initially with a rate that is larger than the population growth rate, and different scenarios are possible. Due to the additional death rate $\sigma$ in the infectious compartment, the population will go on growing but with a lower rate, or the population will become a sub-critical branching process and thus have a decreasing size. In the latter case the population vanishes in the deterministic model while in the stochastic one, what happens when the numbers become low is random. Both the disease and the population could die out, or the population could grow again after the extinction of the epidemic.

We have illustrated and validated the theoretical results by simulations. These simulations show the similarities and the differences between the stochastic model and the corresponding deterministic model. If $R_0>1$, then in the deterministic model, the epidemic will invade the population surely; while in the stochastic model, with a positive probability $\pi$ the disease vanishes. When $R_2<1$, the population vanishes surely in the deterministic model; whereas in the stochastic model, it can regrow exponentially after the extinction of the epidemic. One need to remember that an epidemic is always a stochastic process, and the deterministic model fits only when we have a large community with a large number of infectious individuals \cite{TomB}.

For some diseases (e.g. influenza, measles), the susceptibility and the infectiousness vary with the age of the individuals. Therefore it is more realistic to consider an heterogeneous population as in \cite{Miller}. Furthermore adding seasonal forcing will increase realism for seasonal diseases. For other diseases like sexually transmitted diseases (STD), a close and/or long contact is required for transmission. Then the dynamics of the epidemic are linked to the special network in the host population \cite{Brauer, Frank, Miller}.  Due to the development of the migration of populations, an epidemic starting in one location can be exported to another one quickly, then a meta-population model  is convenient to find the conditions for a global health security \cite{Ball, Colizza}. Under the leadership of the World Health Organization (WHO) many policies (vaccination, medication, quarantine, etc.) are implemented to prevent major outbreak of epidemics. Nevertheless the infectious diseases remain a serious threat for Humanity. Another step towards realism is hence to extend this model by adding vaccination \cite{TomB, Greenhalgh} and treatment \cite{Kumar}. This should allow to find the optimal control strategy to realize the herd immunity, that is to prevent the major outbreaks of infectious diseases.
\section*{Acknowledgments}
We are grateful  to the  International Science Program (ISP) of Uppsala University  and the Swedish International Development Agency (SIDA) for their financial support.

\end{document}